\documentclass[lettersize,journal]{IEEEtran}
\usepackage{amsmath,amsfonts}
\usepackage{algorithmic}
\usepackage{algorithm}
\usepackage{array}
\usepackage{textcomp}
\usepackage{stfloats}
\usepackage{url}
\usepackage{verbatim}
\usepackage{cite}
\usepackage{sharina}
\usepackage{setspace}
\usepackage{amsthm}
\usepackage{amssymb,amsmath,latexsym}
\usepackage{graphicx,subfigure}
\usepackage{color,cite,times}
\usepackage{epstopdf}
\usepackage{multirow}
\usepackage{array}
\usepackage{booktabs}
\newtheorem{assumption}{Assumption}

\newtheorem{theorem}{\textbf{Theorem}}

\newtheorem{proposition}{\textbf{Proposition}}
\newtheorem{remark}{\textbf{Remark}}
\allowdisplaybreaks[4]

\begin{document}
	
	\title{
	 Decentralized Federated  Learning with Distributed Aggregation Weight Optimization}
	\author{		Zhiyuan Zhai, 
		Xiaojun Yuan, \IEEEmembership{Senior Member, IEEE}, Xin Wang, \IEEEmembership{Fellow, IEEE}, and Geoffrey Ye Li, \IEEEmembership{Fellow, IEEE}
	}
	
	\maketitle

\begin{abstract}
%Federated learning (FL) has emerged as an attractive alternative for centralized machine learning (ML),  which enables edge devices to collaboratively train a shared learning model. However, FL's heavy reliance on a central parameter server (PS) leads to issues of communication congestion and reduced fault tolerance. 
Decentralized federated learning (DFL) is an emerging paradigm 
to enable   edge devices  collaboratively training a learning model  using a  device-to-device (D2D) communication manner without the coordination of a parameter server (PS). 
%Communication reliability is a critical perspective of DFL due to the extensive message exchanges. 
%Existing DFL approaches  always assume error-free  communication among devices whereas the real-world communication systems are prone to  failures.
Aggregation weights, also known as mixing weights, are crucial  in DFL process,  and  impact the  learning efficiency and  accuracy.
%In real-world communication networks, due to the varying  conditions of edge devices, the aggregation weights design of DFL must consider the reliability of D2D communication links. 
%However, the aggregation weight design is particularly challenging when deploying DFL in  communication networks.
%Since each edge device only knows its own local  information, 
Conventional design  relies on a so-called central entity to collect  all local information and conduct  system optimization to obtain  appropriate weights.  In this paper, we develop a distributed aggregation weight optimization algorithm to align with the   decentralized nature of DFL.
We analyze convergence  by quantitatively capturing the impact of   the  aggregation weights over decentralized communication networks. Based on the analysis, we then formulate a learning performance optimization problem by designing the aggregation weights  to minimize the derived convergence bound. The optimization problem is further  transformed as an eigenvalue optimization problem and solved by our proposed  subgradient-based algorithm in a  distributed fashion. In our algorithm, 
edge devices  only need local information  to obtain the optimal aggregation weights through  local (D2D) communications, just like the  learning itself. Therefore, the optimization, communication, and learning process can be all conducted in a distributed fashion, which leads to a genuinely  distributed DFL system.
Numerical results demonstrate the superiority of the proposed algorithm in practical  DFL deployment.

\end{abstract}


\section{Introduction}
{Due to the unprecedented increase in local data generated by  edge devices, there is a rising trend in developing deep learning applications at the network edge, which  span many research areas, including image recognition \cite{he2016deep} and natural language processing \cite{young2018recent}. However, traditional machine learning (ML) approaches  need to collect  data from the edge devices for centralized training, which consumes large communication bandwidth  and causes potential privacy concerns. Federated learning (FL)\cite{konevcny2016federated,zhao2018federated,10106001,10040221,liu2024efficient},  a novel machine learning paradigm, is capable of addressing these issues by enabling  edge devices to collaboratively train a global learning model under the coordination of a parameter server (PS) .   In FL, each device independently updates local model, such as model parameters or gradients, using its own datasets. These updates are then transmitted to the central PS, where the aggregated model is computed and broadcast back to the edge devices.}

One important limitation of FL is its heavy dependence on the central PS for parameter aggregation, which  leads to huge communication overheads and decreases fault tolerance.  Furthermore, in scenarios, such as autonomous robotics and collaborative driving \cite{savazzi2021opportunities},  FL may be unrealistic due to the lack of a central PS.
Decentralized federated learning (DFL)  is a promising alternative to tackle these limitations. In DFL, the dependence on the  PS is alleviated by allowing each device to maintain and refine  its own local model, with model updates exchanged through device-to-device (D2D) communications.  The concept of DFL is inherited from decentralized optimization, which can be dated back to the 1980s \cite{tsitsiklis1984problems}, with the foundational algorithms, such as  dual averaging \cite{agarwal2010distributed}, alternating direction method of multipliers (ADMM) \cite{iutzeler2015explicit}, and gradient descent \cite{nedic2009distributed} being pivotal.
Recently, decentralized stochastic gradient descent (DSGD) \cite{sundhar2010distributed}, \cite{amiri2020machine} has appeared as an innovative algorithm for tackling large-scale decentralized optimization challenges. DSGD guarantees optimal convergence under assumptions on convexity, gradient function, and network connectivity. This approach has been further adapted  to tackle various network configurations. 
For instance, the technique in \cite{basu2019qsparse} combines  quantization,  sparsification, and local computations to mitigate communication overhead. Moreover, the algorithm in \cite{koloskova2020unified} incorporates local SGD updates, synchronous communication,  and pairwise gossip  in dynamically changing network topologies.

%Despite the significant promise of DFL, the deployment of DFL faces some problems. Since  each device  receives multiple models through  D2D communication,  a critical task in DFL is  to design  appropriate weights to aggregate the received models.
%However,  this  aggregation weight deign is  challenging.
%Due to the routing constraints, the network that connects the edge devices always has various topologies. When there are a large number of devices, it is difficult to find a general aggregation weight  that optimize the learning efficiency.
%Furthermore, in the real-world communication networks, the imperfect  conditions, such as link fading, additive noise  and limited  resources make  D2D communication links prone to transmission failure. 
%In such scenarios,  some links may receive erroneous data or even fail to receive data altogether.  Hence   the  difference of link reliability impacts the design of aggregation weights.

Aggregation weights, also known as mixing weights, have  significant impacts on the learning performance of DFL. Since each device receives models from many other devices through D2D communication,  these models must be aggregated with appropriate weight for training. When deploying DFL,  aggregation weights are crucial  to improve the training efficiency and learning accuracy \cite{yuan2024decentralized}, \cite{kairouz2021advances}. However, in a decentralized network, where DFL is deployed,   aggregation weight design becomes very challenging due to the complexity of the network structure and the varying quality of D2D communication links. Many existing works,  e.g.,
 \cite{9154332,10506083,9716792,9916128,9783194,roy2019braintorrent}, have  addressed this issue.  
% Take some as examples, the authors of \cite{9154332} consider a wireless network where  devices share a common fading  link for the deployment of DFL. They determine the aggregation weights using a standard choice  based on the node degree of a device in the network graph representation. The authors in \cite{10506083} investigate the design of DFL under MIMO noisy link. Based on the convergence analysis considering communication errors, they propose a communication and learning co-design  algorithm to optimize the aggregation weights.  For a lightweight communication network using the user datagram protocol (UDP), the authors of \cite{9716792} develop a robust decentralized stochastic gradient descent method and optimize the aggregation weights based on the quality of the communication links. 
 However, all of them    rely on the centralized approach and involve the collection of
 the  local information from all  devices.  In these works, it is generally assumed that there exists a central  entity (server or monitor) to gather  the information from devices and link statistics,   then optimize the aggregation weights accordingly. 
This  approach fundamentally contradicts the spirit of decentralization in DFL design \cite{beltran2023decentralized}. 
%A natural question arises: if the entity has the ability  to collect information  and conducts optimization, why not  use it as  PS to perform  traditional FL or   directly use it to  collect the training data for centralized  ML. In such scenarios, conventional FL or ML settings can be more  straightforward and  effective. 
%% If  even DFL  can be differentiated based on the degree of centralization,  then there is no need for this paradigm to exist altogether.
%Moreover, this centralized  approach is also impractical to implement. It is widely known that
%decentralized \emph{ad hoc}  networks are characterized by their complex and dynamic environments, typically situated at a long distance from core networks, with each edge device having limited resources \cite{Bose1999,Leong2005}. Therefore,  finding a viable centralized entity is   infeasible in this situation.

In this paper, we investigate  distributed aggregation weight optimization,  where edge devices  with  local information  cooperatively optimize aggregation weights   through D2D communications, just like  the DFL  learning process. By using this design, the entire optimization, communication, and learning process can be conducted in a  distributed manner, which forms a fully decentralized DFL system. Specifically, we consider a DFL system where devices exchanges models through  D2D communication links and the quality of links are characterized by the link reliability metric. For this scenario, we conduct a rigorous convergence analysis to quantitatively reveal the impact of aggregation weights on the learning performance over the communication networks. Based on this analysis,  aggregation weights  are optimized to minimize the convergence bound, which is a non-convex eigenvalue optimization problem. In general, such a problem is hard to be solved in a distributed fashion. However, by exploiting the problem structure, we are able to develop a distributed subgradient-based algorithm to solve the problem  over  the  communication network. Simulation results confirm our  convergence analysis and  demonstrate the  superiority of the proposed distributed algorithm. Furthermore,  the proposed distributed algorithm can achieve similar performance as  centralized method.

The rest of this paper is outlined as follows. Section II details the DFL learning model and the modeling of communication quality. Section III introduces the preliminary assumptions and derives a  convergence bound. In Section IV, we formulate the performance optimization problem and propose a distributed  algorithm to determine aggregation weights. Section V shows the simulation results, and the conclusion  is provided in Section VI.

\textit{Notations:}
We use the notation $[M]$ to denote the set $\{i | 1 \leq i \leq M\}$ and use $\mathbb{R}$  to denote the sets of real numbers. Scalars are represented by regular letters, vectors by bold lowercase letters, and matrices by bold capital letters. Symbol $\odot$ denotes the Hadamard product (also known as the element-wise product) of two matrices of the same dimensions and  $(\cdot)^\mathrm{T}$ indicates the  transpose operator. Symbol $\text{Diag}(\mathbf{A})$ converts a square matrix $\mathbf{A}$ into a diagonal matrix  by retaining the diagonal elements  and setting all off-diagonal elements to zero. Symbol $\lambda_i(\cdot)$ denotes the $i$-th largest eigenvalue of a matrix.  The $i$-th entry of a vector $\xx$ is denoted as $x_i$, and the $(i,j)$-th entry of a matrix $\mX$ is denoted as $x_{ij}$.  The  $l_2$-norm for vector and Frobenius norm $\norm{\cdot}_F$ for matrix are both denoted as $\norm{\cdot}$. The expectation operator is denoted by $\E$.  $\1$ is used to denote  a vector with all elements equal to $1$ of appropriate dimension. The trace of a square matrix is denoted by $\Tr(\cdot)$. The identity matrix is denoted by $\mathbf{I}$.  The gradient of  function $f$ is denoted as $\nabla f(\cdot)$.

	\section{System Model}\label{sysmodel}
	We first  describe  the decentralized federated learning (DFL) model.
	 In the DFL system,  a set of $M$ devices collaboratively conduct the training of a machine learning (ML) model, where the  common objective  is to minimize  an empirical loss function,  given by
		\begin{align}
		\  f({\bf{x}} )=\frac{1}{M} \sum_{i=1}^M f_i({\bf{x}} ), \  \label{eq:f}
	\end{align}where \(\mathbf{x} \in \mathbb{R}^D\) represents the model  and \(D\) denotes the dimensionality of the parameter space. For an arbitrary device \(i\), the local loss function \(f_i : \mathbb{R}^D \rightarrow \mathbb{R}\), is formulated as
	\begin{equation}
		f_i(\mathbf{x}) := \mathbb{E}_{\xi_i \in \mathcal{D}_i} F(\mathbf{x}, \xi_i),
	\end{equation}
where \(\mathcal{D}_i\) is the  local dataset on device \(i\) and \(F(\mathbf{x}, \xi_i)\) is the loss function with respect to samples \(\xi_i\).
	\begin{table}[]
		\centering{
			\caption{{Summary of Main Notations}}
			\footnotesize
			\begin{tabular}{@{}l|l@{}}
				\toprule
				\textbf{Notation} & \textbf{Definition} \\ 
				\midrule
%				$f(\xx)$ & Common  loss function \\ 
%				$f_i(\xx)$ & Loss function with respect to device $i$\\ 
%				$F(\xx,\xi_i)$ & Loss function with respect to model $\xx$ and data samples $\xi_i$ \\ 
				$\mP$; $p_{ij}$ & Link reliability matrix; $(i,j)$-th element of $\mP$ \\ 
				$\mS^{(t)}$; $s^{(t)}_{ij}$& Transmission  matrix at round $t$; $(i,j)$-th element of $\mS^{(t)}$ \\ 
				$\mW$; $w_{ij}$ & Aggregation weight matrix; $(i,j)$-th element of $\mW$ \\ 
				$\widehat{\mW}^{(t)}$; $\widehat{w}^{(t)}_{ij}$ & Mixing matrix at round $t$; $(i,j)$-th element of $\widehat{\mW}^{(t)}$\\ 
				$\overline{\mW}$; $\overline{w}_{ij}$ & Expectation of 	$\widehat{\mW}^{(t)}$; $(i,j)$-th element of $\overline{\mW}$\\ 
				$\overline{\mW^2}$; $\overline{w^2}_{ij}$ & Expectation of 	$(\widehat{\mW}^{(t)})^2$; $(i,j)$-th element of $\overline{\mW^2}$\\ 
				$\rho(\overline{\mW})$ & Second-largest  eigenvalue (in magnitude) of ${\overline{\mW}}$ \\ 
				${\rho(\overline{\mW^2})}$ & Second-largest  eigenvalue (in magnitude) of ${\overline{\mW^2}}$ \\ 
				\bottomrule
			\end{tabular}
		}
	\end{table}

In such a system, the devices compute the local model by minimizing their  local loss function,  and exchange the acquired model  through  D2D communications.
% The presence and operation of a communication link between device \(i\) and device \(j\) are indicated by the function \(e_{ij}\), where \(e_{ij} = 1\) means the communication link exists, and \(e_{ij} = 0\) signifies its absence. Due to the link reciprocity, we have $e_{ij}=e_{ji}$. This configuration allows the network's communication topology, instrumental in model parameter exchanges, to be modeled as an undirected graph \(\mathcal{G} = (\mathcal{V}, \mathcal{E})\), where \(\mathcal{V}\) represents the set of devices and \(\mathcal{E}\) embodies the set of  communication links, i.e., \(\mathcal{E} = \{e_{ij} | e_{ij} = 1, \forall i, j\}\). A device \(i\) is considered to be a neighbor of device \(j\) if \(e_{ij} = 1\).
Most of recent decentralized learning frameworks \cite{sundhar2010distributed,amiri2020machine,basu2019qsparse ,koloskova2020unified} also assume  stable and reliable communication links. However, in practical  networks, the communication system is prone to transmission distortion or failure. 
The practical conditions, such as channel fading, additive noise, path loss,  and limited  resources,   make the D2D  link  unreliable.
{During each communication round,  only a subset of the links can achieve successful communication.
The quality of link can be quantified   by  a link reliability matrix \(\mP  \in \mathbb{R}^{M \times M}\), which is  invariant during the whole training process\footnote{Note that our analytical framework can be directly applied to the situation where the transmission probability is dynamic, i.e., $\mP$ changes over training rounds. In this paper, we focus on the static case to simplify our analysis.}.   In this matrix, the $(i,j)$-th element, \(p_{ij}\), denotes the probability of  successful  transmission from the \(i\)-th device to the \(j\)-th device\footnote{This link reliability model can be applied to  decentralized network with any topology. For a specific   topology where  certain devices cannot communicate, we  just  need to set the corresponding reliability $p_{ij}=0$. This probability model allows us to capture the stochastic and unreliable nature of D2D communication while keeping the optimization framework general and independent of a specific channel model \cite{10878812}.
}.  }
Since a device does not need to communicate with itself, we have \( p_{ii} = 0,\forall i \in [M] \)\footnote{We can also assume $p_{ii}=1$ and rewrite \eqref{model_exchange} as $	\xx_i^{(t+\frac{1}{2})}=\sum_{j=1}^{M}w_{ij}\hat{\rr}_{j \to i}^{(t)}$ correspondingly. This does not hurt our theoretical analysis.}.

% It is presupposed that the communication topology remains unchanged throughout the whole training process.
	\begin{figure}[htbp]
	\centering
	\includegraphics[width=0.4\textwidth]{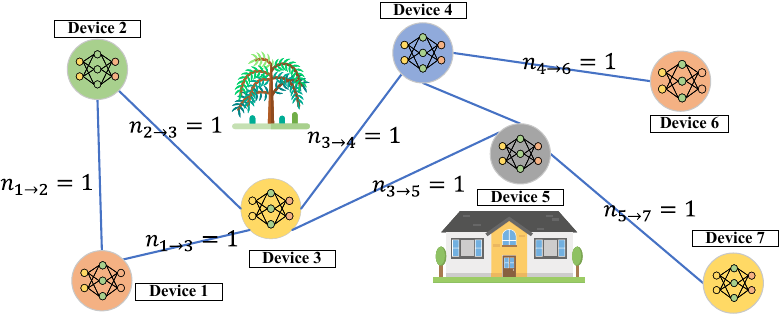}
	\caption{An example of the  links with successful transmission.}
	\label{fig_system_model}
\end{figure}

When the \(i\)-th device broadcasts its model, i.e., \(\xx_i \in \mathbb{R}^D\), to the rest of network, this model is  received by other devices with some random distortion. 
%We assume that the  link remains invariant during each DFL communication round. 
The reception of $\xx_i$ at the \(j\)-th device, denoted by \(\rr_{i \to j}\), is  expressed as \(\rr_{i \to j} = n_{i \to j} \xx_i\), where \(n_{i \to j} = 1\) if  the model  of the \(i\)-th device is successfully received by the \(j\)-th device, and \(n_{i \to j}= 0\) otherwise. Without loss of generality, we assume the reciprocity of the transmission link, that is \( n_{i \to j} = n_{j \to i}, \forall i,j \). 
An example of the links that achieve successful  communication in one round is shown in Fig.~\ref{fig_system_model}, where the blue line represents the links with successful transmission.

%\footnote{\textcolor[rgb]{0,0,1}{ This paper primarily considers device-to-device link reliability factor, which is a key structural distinction between DFL and FL. For other factors, such as client data quantity and data distribution, we can draw upon methodologies from traditional FL and apply them to the design of aggregation weights in DFL.}}.

We are now ready to introduce the training process  of the DFL system. We adopt the stochastic gradient descent (SGD) technique \cite{zinkevich2010parallelized} for  local model training.  The models of each  device are updated in an iterative manner at each training round. In the 
$t$-th training round, the training process consists of the following steps.

\begin{itemize}
	\item \textit{Local  computation}: Each device, say device $i$, computes its local stochastic gradient $\nabla F(\mathbf{x}_i^{(t)}, \xi_i^{(t)})$ by  randomly selecting  data $\xi_i^{(t)}$ from the  local dataset $\mathcal{D}_i$, where $\mathbf{x}_i^{(t)}$ represents the model  of device $i$ at the $t$-th training round.
	\item \textit{Model exchange}: Each device communicates with others to exchange the model parameters. Specifically, each device broadcasts the  model to the rest of the network.
	If transmission fails in certain communication links, the receiving device cannot obtain the correct model. To mitigate this issue, the receiving device will substitute this lost value with its own model.
%	 Take an example, the message from device $j$ to device $i$ is $\rr_{i \to j}^{(t)}=n_{i \to j}^{(t)}\xx_j^{(t)}$.
	 Therefore,  the received model of device $i$ from device $j$ is expressed as
	 $\hat{\rr}_{j \to i}^{(t)}=\rr_{i \to j}^{(t)}+(1-n_{i \to j}^{(t)})\xx_i^{(t)}$.  After that, each device calculates the weighted average of models as
	\begin{align}
		\xx_i^{(t+\frac{1}{2})}&=w_{ii}\xx_i^{(t)}+\sum_{j=1,j\neq i}^{M}w_{ij}\hat{\rr}_{j \to i}^{(t)}\notag\\
			&=\xx_i^{(t)}+\sum_{j=1,j\neq i}^{M}w_{ij}n_{j \to i}^{(t)}(\xx_j^{(t)} -\xx_i^{(t)}),\label{model_exchange}
	\end{align}
where \( w_{ij} \) represents the aggregation weight of the model from device \( j \) when device \( i \) performs model aggregation, and $\mathbf{x}_{i}^{(t + \frac{1}{2})}$ denotes the  aggregated model at device $i$ at the $t$-th training round. 
	\item \textit{Local  update}: Using the aggregated model,  ${\mathbf{x}}_{i}^{(t + \frac{1}{2})}$,  each device updates its local model as
	\begin{align}
		\mathbf{x}_{i}^{(t+1)} = {\mathbf{x}}_{i}^{(t+\frac{1}{2})} - \lambda \nabla F(\mathbf{x}_i^{(t)}, \xi_i^{(t)}), \quad \forall i \in [M],\label{update_vector}
	\end{align}
	where $\lambda \in \mathbb{R}$ is the learning rate.
\end{itemize}

{We  view  the above DFL training process   from a global perspective.
Define the concatenation of devices' models and stochastic gradients at round $t$ as
${\mathbf{X}}^{(t)}\triangleq \left[{\xx}_{1}^{(t)},\dots, {\xx}_{M}^{(t)}~\right]$ and $\partial F(\mX^{(t)},{\boldsymbol{\xi}}^{(t)}) \triangleq\left[\nabla F(\xx^{(t)}_1,\mathbf{\xi}^{(t)}_1),\dots,\nabla F(\xx^{(t)}_M,\mathbf{\xi}^{(t)}_M)\right]$, respectively.  Then,  iterative update formula \eqref{update_vector} can be rewritten in matrix form as
\begin{align}
	\mX^{(t+1)}=\mX^{(t)}\widehat{\mW}^{(t)}-\lambda \partial F(\mX^{(t)},{\boldsymbol{\xi}}^{(t)}),\label{mat_up}
\end{align}
where   $\widehat{\mW}^{(t)} \in \mathbb{R}^{M\times M}$ is the  mixing matrix in the $t$-th round and is given by $\widehat{\mW}^{(t)} =\mI +\mW\odot\mS^{(t)} -\text{Diag}(\mW\mS^{(t)})$,  and $\mS^{(t)} \in \mathbb{R}^{M\times M}$ is the matrix representing the success of transmission at round $t$. If device $i$ successfully transmits its model  to device $j$, then $s^{(t)}_{ij}=1$;  otherwise, $s^{(t)}_{ij}=0$. In particular, we have $s^{(t)}_{ii}=0,\forall i$, and $s^{(t)}_{ij}=s^{(t)}_{ji},\forall i,j$. 

The randomness in the considered DFL system  arises from two aspects: one is the randomness of the communication link, and the other is the randomness of the training samples. At  the \( t \)-th iteration, given the current model parameters \( \mathbf{X}^{(t)} \) and the training samples \( \boldsymbol{\xi}^{(t)} \), the expectation of  \eqref{mat_up} is
\begin{align}
	\mathbb{E}_{(\cdot|\mX^{(t)},{\boldsymbol{\xi}}^{(t)})}\{\mX^{(t+1)}\}=\mX^{(t)}\overline{\mW}-\lambda \partial F(\mX^{(t)},{\boldsymbol{\xi}}^{(t)}),
\end{align}
where  $\overline{\mW}$ is the expectation of $\widehat{\mW}^{(t)}$, i.e., $\overline{\mW}=\mathbb{E}\{\widehat{\mW}^{(t)}\}$, with its $(i,j)$-th entry  given by
}
\begin{align}\label{77}
	\overline{w}_{ij} &= \begin{cases}
		w_{ij} p_{ij},  \quad \quad \quad \quad \quad \quad  & i \neq j\\
		1 - \sum_{l = 1, l\neq i}^M w_{i,l} p_{i,l}. \quad & i = j
	\end{cases} 
\end{align}
\section{Convergence analysis}
Our convergence analysis is based on the following assumptions.

	\begin{assumption}\label{as1}{
	\rm{(Symmetric  communication)} 
	We assume that the communication reliability from device $i$ to device $j$ is equal to that from device $j$ to device $i$.  That is,  the  link reliability matrix, $\mP$,  is a symmetric matrix, i.e., $\mP=\mP^{\mathrm{T}}$.}
\end{assumption}
	\begin{assumption}\label{as2}
	\rm{(Independent  connections)}
The transmission reliability of different communication links ($p_{ij}$ and $p_{ji}$ are considered to  be the reliability of the same communication link) are independent.

\end{assumption}
	\begin{assumption}\label{as3}
	\rm{(Aggregation weights)}
 Assume $\mW$ is a symmetric doubly stochastic matrix, i.e.,  $\mathbf{W}^{\mathrm{T}}=\mathbf{W}$, $ \mathbf{W}\1=\1$.
 Define $\overline{\mW^2}=\E\{(\widehat{\mW}^{(t)})^2\}$ and	${\rho(\overline{\mW^2})} \triangleq \max\{\lambda_2(\overline{\mW^2}), -\lambda_M(\overline{\mW^2})\}$. We  assume ${\rho(\overline{\mW^2})} < 1$.
\end{assumption}

	\begin{assumption}\label{as4}
	\rm{(Smoothness)} Loss functions $f_1,\dots,f_M$ are all differentiable and the their  gradients $\nabla f_1(\cdot),\dots,f_M(\cdot)$ are Lipschitz continuous with parameter $\omega$, i.e.,
	\begin{align}
		\norm{\nabla f_i(\xx)-\nabla f_i(\yy)}\leq \omega \norm{\xx-\yy}, \forall \xx,\yy \in \mathbb{R}^{D},\forall i \in [M].\notag
	\end{align}
\end{assumption}

	\begin{assumption}\label{as5}{
		\rm{(Bounded variance)} The variances of stochastic gradient $\E{\norm{\nabla F(\xx, \xi_i) - \nabla f_i(\xx)}}^2$ and $	\E{\norm{\nabla f_i(\xx) - \nabla f(\xx)}}^2$ are bounded, i.e.,
	\begin{align}
		\E\norm{\nabla F(\xx, \xi_i) - \nabla f_i(\xx)}^2 &\leq \alpha^2\,, \forall\xx \in \mathbb{R}^D\,,\forall i \in [M],\notag\\
		\E\norm{\nabla f_i(\xx) - \nabla f(\xx)}^2 &\leq  \beta^2\,,   \forall \xx \in \mathbb{R}^D,\notag
	\end{align}
where $\alpha^2$ denotes the upper limit of the variance  of stochastic gradients among  devices, and $\beta^2$ denotes the upper limit of the divergence of data distributions among  devices. The expectation is taken over the randomness of local data sampling~$\xi_i$ in the first inequality 
and over the random selection of device index~$i$ (i.e., $i\!\sim\!\mathcal{U}([M])$) in the second inequality.}
\end{assumption}

Assumptions \ref{as1} and \ref{as2} are  related to the communication networks. Channel reciprocity ensures equal reliability for bidirectional transmissions over the same link, and hence Assumption  \ref{as1} holds. Furthermore,  Assumption  \ref{as2} holds as long as the distance between devices significantly exceeds the carrier wavelength \footnote{{In some specific network structures, the link reliability matrix $P$ may have dependent elements. Extending the analytical framework to account for these dependencies is beyond the scope of this paper but remains an important direction for future work.}} \cite{9716792}.  

{
Assumptions \ref{as3}-\ref{as5}  are widely adopted in research concerning decentralized  stochastic optimization  algorithms, e.g., \cite{koloskova2019decentralized}, \cite{wang2018cooperative}, and \cite{li2021decentralized}. 
%Assumption \ref{as3} is the characteristics of the aggregation weight matrix. 
 From  definitions,  matrices $\widehat{\mW}^{(t)}$, $\overline{\mW}$ and $\overline{\mW^2}$ are also shown to be symmetric doubly stochastic if  aggregation weight matrix $\mW$ is symmetric doubly stochastic.
  It is  known that for a doubly stochastic matrix ${\mW}$, $\lambda_1({\mW})= 1$ and $|\lambda_i({\mW})| \leq 1, \forall i$. Assumption \ref{as3} requires that for all $i \neq 1$,  $|\lambda_i({\mW})|$ must be strictly less than $1$. This assumption is always true if the underlying graph of the communication network is connected and non-bipartite \cite{west2001introduction}.  Assumption \ref{as4} is about the Lipschitz continuity property of the loss function.  Assumption  \ref{as5} guarantees that there is a limited disparity between the gradient of  local sample-dependent loss function $\nabla F(\xx, \xi_i)$ and the gradient of  overall loss function $\nabla f(\xx)$. In practice, the assumption holds locally within the region visited by the algorithm, as the iterates remain bounded due to step-size control, Lipschitz continuity, and regularization.
Based on the above assumptions, we can derive the following convergence theorem, proved in Appendix \ref{appb}.}
\begin{theorem}\label{pro2}
	 {Under Assumptions \ref{as1}-\ref{as5}, with   $	\lambda < \frac{1 - \sqrt{\rho(\overline{\mathbf{W}}^2)}}{6 \sqrt{M} \omega}$, we have the following  ergodic convergence bound }
	\begin{align}
		&\frac{1}{T}\sum_{t=0}^{T-1}\E\norm{\nabla f\left(\frac{\mX^{(t)}\1}{M}\right)}^2\leq\frac{1}{\left({1}/{2}-{9}G(\overline{\mW^2})\right)}\notag\\&\quad\quad\times\left(\frac{f_0-f^*}{\lambda T}+\frac{\lambda\omega\alpha^2}{2M}
		+{\alpha^2}G(\overline{\mW^2})+{9\beta^2}G(\overline{\mW^2})\right)\label{conver_bound}
	\end{align}
where   $G(\overline{\mW^2})=\frac{M\lambda^2\omega^2}{(1-\sqrt{{\rho(\overline{\mW^2})}})^2-18M\lambda^2\omega^2}$, $f_0$ (or $f^*$) is initial (or optimal) value of the loss function, and the expectation is taken over the stochastic transmission and the  randomness of data sampling.
\end{theorem}

{
\begin{remark}
From Theorem~1, the ergodic bound consists of a decaying term proportional to 
$1/T$ and several constant terms related to $\lambda$, $M$, $\alpha$, $\beta$, 
and $G(\overline{\mW^2})$. 
Therefore, the algorithm exhibits an $\mathcal{O}(1/T)$ {transient convergence rate} 
and converges to a {neighborhood of the optimum} determined by these constant terms. 
The neighborhood becomes smaller with larger $M$, smaller learning rate $\lambda$, 
lower gradient variance $(\alpha,\beta)$, and better network mixing 
(smaller $G(\overline{\mW^2})$).
\end{remark}}
\begin{proposition}\label{pro3}
The convergence bound in \eqref{conver_bound} is a monotonically increasing function with respect to ${\rho(\overline{\mW^2})}$.
\end{proposition}
\begin{proof}
	The convergence bound \eqref{conver_bound} can  be abbreviated as $f(G(\overline{\mW^2}))=\frac{A+BG(\overline{\mW^2})}{1/2-9G(\overline{\mW^2})}$, where  $A=\frac{f_0-f^*}{\lambda T}+\frac{\lambda\omega\alpha^2}{2M}, B =\alpha^2+9\beta^2$. The  derivative of $f$ with respect to $G(\overline{\mW^2})$ is $f'(G(\overline{\mW^2}))=\frac{1/2+9A}{(1/2-9G(\overline{\mW^2}))^2}>0$, which means that $f(G(\overline{\mW^2}))$ monotonically increases with respect to $G(\overline{\mW^2})$. Furthermore, since  $0\leq {\rho(\overline{\mW^2})}<1$,  $G(\overline{\mW^2})$ also monotonically increases with respect to ${\rho(\overline{\mW^2})}$. Hence, the convergence bound in \eqref{conver_bound} is  monotonically increasing with respect to ${\rho(\overline{\mW^2})}$.
\end{proof}

From Proposition \ref{pro3},  aggregation weights \( \mW \) should be designed to minimize the value of ${\rho(\overline{\mW^2})}$ as much as possible, which provides a guideline for aggregation weight optimization.

\section{Distributed Weight Optimization}\label{DO}
In order to improve the learning performance of  DFL, we shall design  aggregation weight matrix $\mW$ to minimize an objective function obtained from the convergence bound in \eqref{conver_bound}. We will address this issue in this section.
\subsection{Motivation of Distributed Optimization}\label{nece}

 Aggregation weights dictate how well models from various devices are combined, which in turn affects the system's capability to generalize and learn effectively. 
 From Proposition \ref{pro3},  the   aggregation weights should be chosen to minimize  ${\rho(\overline{\mW^2})}$ in order for the DFL system to operate at its optimal performance.

In previous works \cite{9154332,10506083,9716792,9916128,9783194,roy2019braintorrent}, the  aggregation weights are optimized based  on a centralized approach, which involves the  global information from all devices within the network and typically requires a central entity to collect local device states and link conditions. Since devices are often limited to device-to-device  communications, distributed aggregation weight optimization is more practical than centralized approach.

\subsection{Surrogate Objective Function}\label{surro}
From Proposition \ref{pro3},  the convergence rate  decreases monotonically with the increase of  ${\rho(\overline{\mW^2})}$. To enhance the learning performance, we need  to minimize ${\rho(\overline{\mW^2})}$, which 
however, is a nonlinear function of the second-order statistics, $\overline{\mW^2}=\E\{(\widehat{\mW}^{(t)})^{\mathrm{T}}(\widehat{\mW}^{(t)})\}$. To design an effective decentralized algorithm,  we  replace ${\rho(\overline{\mW^2})}$ by a tractable surrogate objective function  .

Our design is based on the following proposition, which is proved in Appendix \ref{app_b}.
\begin{proposition}\label{pro33}
Assume a  large-scale  system where the number of  devices approaches infinity,	i.e., $M\to\infty$, $\rho(\overline{\mW})=\max\{\lambda_2(\overline{\mW}),-\lambda_M(\overline{\mW})\}$ can serve as a surrogate objective function of ${\rho(\overline{\mW^2})}$.
\end{proposition}
%\begin{proof}
%	Please refer to Appendix \ref{app_b}.
%\end{proof}

Now we are ready to formulate the  DFL learning performance optimization  as follows.
\begin{subequations}\label{problem1}
	\begin{align}
		&\min_{\mW} \quad\rho(\overline{\mW})=\max\{\lambda_2(\overline{\mW}),-\lambda_M(\overline{\mW})\}\\
		& ~~\text{s.t.}\quad \mW^{\mathrm{T}}=\mW,  \mW\1=\1, \mW \in [0,1]^{M\times M}.
	\end{align}
\end{subequations}
Problem \eqref{problem1} is an eigenvalue  optimization problem. The main challenge  comes from the distributed solving restriction. We develop a distributed subgradient-based algorithm to address this problem subsequently.
\subsection{Subgradient Analysis}\label{SA}
Since $\overline{\mW}$ is a doubly stochastic symmetric matrix with   eigenvalue $\lambda_1(\overline{\mW})=1$, problem \eqref{problem1} can be transformed to
\begin{subequations}\label{problem2}
	\begin{align}
		&\min_{\mW} \quad \rho(\overline{\mW})=\max\{\lambda_2(\overline{\mW}),-\lambda_M(\overline{\mW})\}\\
		& ~~\text{s.t.}\quad \mW^{\mathrm{T}}=\mW,  \mW\1=\1, \mW \in [0,1]^{M\times M}.\label{10b}
	\end{align}
\end{subequations}
 We call the second largest (in magnitude) eigenvalue,  $\rho(\overline{\mW})$, the mixing rate of  $\overline{\mW}$.   Since  $\lambda_1(\overline{\mW})=1$, we can express the second largest eigenvalue  as
\begin{align}\label{ss1}
	\lambda_2(\overline{\mW})=\sup \{\uu^{\mathrm{T}}\overline{\mW}\uu~|\norm{\uu}_2\leq 1,\1^{\mathrm{T}}\uu=0\}.
\end{align}
As $\lambda_2(\overline{\mW})$ is a point-wise supremum of a family of linear functions ($\uu^{\mathrm{T}}\overline{\mW}\uu$) of $\overline{\mW}$, it is thus convex\cite[Section 3.2.3]{boyd2004convex}. Similarly, the negative of the smallest eigenvalue $-\lambda_M(\overline{\mW})$ can be expressed as
\begin{align}
	-\lambda_M(\overline{\mW})=\sup \{-\uu^{\mathrm{T}}\overline{\mW}\uu~|\norm{\uu}_2\leq 1\},
\end{align}
which is also convex. Therefore, $\rho(\overline{\mW})=\max\{\lambda_2(\overline{\mW}),-\lambda_M(\overline{\mW})\}$ is the point-wise maximum of two convex functions and hence it is convex.
The subsequent discussion needs the following proposition.
\begin{proposition}\label{propo5}
	A subgradient of $\rho(\overline{\mW})$ is a symmetric matrix $\mG$ that satisfies the following inequality
\begin{align}
	\rho(\overline{\mW}')&\geq \rho(\overline{\mW})+ \mathrm{Tr}~( \mG(\overline{\mW}'-\overline{\mW})),
\end{align}
where $\overline{\mW}'$ is an arbitrary symmetric doubly stochastic matrix, $\langle\cdot,\cdot\rangle$ represents the matrix inner product. When $\rho(\overline{\mW})=\lambda_2(\overline{\mW})$ and $\vv$ is the unit eigenvector corresponding to $\lambda_2(\overline{\mW})$, the subgradient  is given by $\mG=\vv\vv^{\mathrm{T}}$. Similarly, when $\rho(\overline{\mW})=-\lambda_M(\overline{\mW})$ and $\vv$ is a unit eigenvector corresponding to $\lambda_M(\overline{\mW})$, we have  $\mG=-\vv\vv^{\mathrm{T}}$.
\end{proposition}
\begin{proof}
	We first consider the case, $\rho(\overline{\mW})=\lambda_2(\overline{\mW})$, and $\vv$ is the corresponding unit eigenvector. Since $\1$ is the eigenvector for eigenvalue $1$ of matrix $\overline{\mW}$, we have $\vv^{\mathrm{T}}\1=0$. By using the variational characterization of the second largest eigenvalue $\lambda_2$ of matrix $\overline{\mW}$ and $\overline{\mW}'$, we have
	\begin{subequations}
			\begin{align}
			&\rho(\overline{\mW})=\lambda_2(\overline{\mW})=\vv^{\mathrm{T}}\overline{\mW}\vv,\label{14a}\\
			&\rho(\overline{\mW}')\geq\lambda_2(\overline{\mW}')\geq\vv^{\mathrm{T}}\overline{\mW}'\vv.\label{14b}
		\end{align}
	\end{subequations}

Subtracting the two sides of \eqref{14a}  from those of \eqref{14b}, we have
\begin{align}
	\rho(\overline{\mW}')&\geq \rho(\overline{\mW})+\vv^{\mathrm{T}}(\overline{\mW}'-\overline{\mW})\vv,\notag\\
	&=\rho(\overline{\mW})+ \mathrm{Tr}~ (\vv\vv^{\mathrm{T}}(\overline{\mW}'-\overline{\mW})).
\end{align}
Hence, $\mG=\vv\vv^{\mathrm{T}}$ is a subgradient when $\rho(\overline{\mW})=\lambda_2(\overline{\mW})$. Similarly, we can prove  $\mG=-\vv\vv^{\mathrm{T}}$ when $\rho(\overline{\mW})=-\lambda_M(\overline{\mW})$ and $\vv$ is a unit eigenvector corresponding to $\lambda_M(\overline{\mW})$.
\end{proof}

Define  matrices $\mE_{ij}$, with entries ${\mE_{ij}}(i,j)={\mE_{ij}}({j,i})=p_{ij}, {\mE_{ij}}({i,i})={\mE_{ij}}({j,j})=-p_{ij}$, and zero entries everywhere else. 
Therefore, we can recast  optimization problem \eqref{problem2}  as 
\begin{subequations}\label{problem3}
	\begin{align}
		&\min_{\mW} \quad \rho\left(\mI+\frac{1}{2}\sum_{i,j=1}^{M}{w}_{ij}\mE_{ij}\right)\\
		& ~~\text{s.t.}\quad \mW^{\mathrm{T}}=\mW,  \mW\1=\1, \mW \in [0,1]^{M\times M}.
	\end{align}
\end{subequations}

Denote $\mR(\mW)=\mI+\frac{1}{2}\sum_{i,j=1}^{M}{w}_{ij}\mE_{ij}$. In the subgradient method, we need to calculate the subgradient of the objective function $\rho(\mR(\mW))$ for a given feasible $\mW$.
If $\rho(\mR(\mW))=\lambda_2(\mR(\mW))$ and $\vv$ is the corresponding unit eigenvector.
From Proposition \ref{propo5}, we have
\begin{align}
	\lambda_2(\mR(\mW'))\geq \lambda_2(\mR(\mW))+\sum_{i,j=1}^{M}\left(\vv^{\mathrm{T}}\mE_{ij}\vv\right)(w_{ij}'-w_{ij}),
\end{align}
so  subgradient $g(\mW)$ is expressed as
\begin{align}
	g(\mathbf{W}) &= 
	\begin{bmatrix}
		\mathbf{v}^\top \mathbf{E}_{11} \mathbf{v} & \mathbf{v}^\top \mathbf{E}_{12} \mathbf{v} & \cdots & \mathbf{v}^\top \mathbf{E}_{1M} \mathbf{v} \\
		\mathbf{v}^\top \mathbf{E}_{21} \mathbf{v} & \mathbf{v}^\top \mathbf{E}_{22} \mathbf{v} & \cdots & \mathbf{v}^\top \mathbf{E}_{2M} \mathbf{v} \\
		\vdots & \vdots & \ddots & \vdots \\
		\mathbf{v}^\top \mathbf{E}_{M1} \mathbf{v} & \mathbf{v}^\top \mathbf{E}_{M2} \mathbf{v} & \cdots & \mathbf{v}^\top \mathbf{E}_{MM} \mathbf{v}
	\end{bmatrix},
\end{align}
and  subgradient component $g(w_{ij})$ is 
\begin{align}\label{17}
	g(w_{ij})=\vv^{\mathrm{T}}\mE_{ij}\vv=-p_{ij}(v_i-v_j)^2.
\end{align}

Similarly, if $\rho(\mR(\mW))=-\lambda_M(\mR(\mW))$ and $\vv$ is the corresponding unit eigenvector,  subgradient  is given by
\begin{align}
	g(\mathbf{W}) &= 
	\begin{bmatrix}
		-\mathbf{v}^\top \mathbf{E}_{11} \mathbf{v} & -\mathbf{v}^\top \mathbf{E}_{12} \mathbf{v} & \cdots & -\mathbf{v}^\top \mathbf{E}_{1M} \mathbf{v} \\
		-\mathbf{v}^\top \mathbf{E}_{21} \mathbf{v} & -\mathbf{v}^\top \mathbf{E}_{22} \mathbf{v} & \cdots & -\mathbf{v}^\top \mathbf{E}_{2M} \mathbf{v} \\
		\vdots & \vdots & \ddots & \vdots \\
		-\mathbf{v}^\top \mathbf{E}_{M1} \mathbf{v} & -\mathbf{v}^\top \mathbf{E}_{M2} \mathbf{v} & \cdots & -\mathbf{v}^\top \mathbf{E}_{MM} \mathbf{v}
	\end{bmatrix},
\end{align}
with the $(i,j)$-th subgradient component
\begin{align}\label{199}
	g(w_{ij})=-\vv^{\mathrm{T}}\mE_{ij}\vv=p_{ij}(v_i-v_j)^2.
\end{align}
\begin{remark}
	From \eqref{17} and \eqref{199},  for the subgradient of the aggregation weight between devices $i$ and $j$, i.e., $w_{ij}$, we only need to know the link reliability information, i.e., $p_{ij}$ and the $(i,j)$-th components of the unit eigenvector, i.e., $v_i$ and $v_j$. This  implies that if each device, say device $i$, knows its own link reliability $p_{ij},\forall j$  and the eigenvector component $v_{i}$, then the subgradient can be computed in a distributed manner by using only  local information.
\end{remark}
%Since each device knows its own link reliability,  obtain its eigenvalue component. In the next subsection, we introduce a distributed eigenvalue computation method.
\subsection{Distributed Eigenvector Computation}\label{dec}
%We see from Section \ref{SA} that if each device knows its own component of eigenvector $\vv$, the subgradient can be computed  according to \eqref{17} and \eqref{199}.
In this subsection, we    compute the  eigenvector component of  $\overline{\mW}$ for the corresponding device  in a distributed fashion, where  device $i$ is only aware of the $i$-th row of $\overline{\mW}$ and can only communicate  with its neighbors. 
%we mean that each device $i$ only knows the $i$-th row of $\overline{\mW}$ and can only exchange information with its immediate neighbors.
 
The problem of distributed computation of  the top $k$ eigenvectors of a symmetric weighted adjacency matrix of a graph is discussed  in \cite{kempe2004decentralized},  the orthogonal iteration  algorithm in \cite{xu2018convergence} is adopted to the distributed environment for a QR decomposition based approach.

Since  matrix $\overline{\mW}$ is symmetric doubly stochastic with the largest eigenvalue  $1$  and the corresponding eigenvector  $\1$, the orthogonal iterations take on a very simple form as summarized in Algorithm \ref{algOI}. We do not need to calculate any QR decomposition at  device for orthogonalization. 
\begin{algorithm}[htb]
	\caption{Distributed Orthogonal Iterations} 
	\label{algOI} 
	\begin{algorithmic}[1] %这个1 表示每一行都显示数字
		\STATE {\textbf{Initialization:}}  Random chosen vector $\vv(0)$, iteration index $k=1$, and number of iterations $K_{\max}$.
		\FOR{ $k \leq K_{\max}$ }
		\STATE Update $\vv(k+1)=\overline{\mW}\vv(k)$.
		\STATE Orthogonalize $\vv(k+1)=\vv(k+1)-\frac{1}{M}\vv(k+1)^{\mathrm{T}}\1\1$.
		\STATE Scale to vector $\vv(k+1)=\vv(k+1)/\norm{\vv(k+1)}$.
		\ENDFOR
		\ENSURE $\vv(k+1)$.%算法的输出：Output
	\end{algorithmic}
\end{algorithm}
\begin{remark}\label{remark3}
 Algorithm \ref{algOI} can  be performed in a distributed manner.
	 In the initialization step,  vector $\vv(0)$ can be constructed by the way that each device  generates a random scalar component.
	 To obtain the $i$-th element of $\vv(k+1)$ in the multiplication $\vv(k+1)=\overline{\mW}\vv(k)$ step, the $i$-th device can fetch the corresponding components of $\vv(k)$ from its neighbors and aggregate them according to the $i$-th row of $\overline{\mW}$.
	 The orthogonalization and scaling steps involve the operation that sum over each  element  of $\vv(k+1)$. This  can be achieved through the distributed averaging method introduced in \cite{xiao2004fast}. It can also  be achieved by simply  executing step 3 repeatedly, since $\overline{\mW}$ is a doubly stochastic matrix where $\overline{\mW}^\infty=\frac{1}{M}\1\1^{\mathrm{T}}$.
\end{remark}
\begin{remark}
According to  \cite[Section 7.3.1]{golub2013matrix},  obtained eigenvector $\vv(k+1)$ in Algorithm \ref{algOI} is naturally associated with the second largest (in magnitude) eigenvalue, i.e., $\rho(\overline{\mW})$. Therefore, $\vv(k+1)$ can be directly applied to the subgradient evaluation in Section \ref{SA} (without the need to identify whether $\rho(\mR(\mW))=\lambda_2(\mR(\mW))$ or $\rho(\mR(\mW))=-\lambda_M(\mR(\mW))$).
\end{remark}

\subsection{Subgradient Projection Algorithm}
In this subsection, we develop a distributed  subgradient projection algorithm to solve \eqref{problem2} based on the analysis in Section  \ref{SA} and  \ref{dec}. The details of this algorithm is given in Algorithm \ref{algpsa}.
\begin{algorithm}[htb]
	\caption{Subgradient Projection Algorithm for \eqref{problem2}} 
	\label{algpsa} 
	\begin{algorithmic}[1] %这个1 表示每一行都显示数字
		\STATE {\textbf{Initialization:}} Channel reliability matrix $\mP$, a feasible matrix $\mW(0)$,  iteration index $n=1$, and number for iterations $J_{\max}$.
		\FOR{ $n \leq J_{\max}$ }
		\STATE Compute unit eigenvector $\vv(n)$ based on Algorithm \ref{algOI}.
		\STATE Compute subgradient g($\mW(n)$) based on  \eqref{17} and \eqref{199}.
		\STATE Update $\mW$ as $\mW(n+1)=\mW(n)-\gamma_n g(\mW(n))$.
		\STATE  Project 	$\mW$ onto the feasible set according to  \eqref{222} and \eqref{23}.
		\ENDFOR
		\ENSURE $\mW$.
	\end{algorithmic}
\end{algorithm}

We now show how to obtain the projected results with a distributed method. Since $\mW$ is constrained to be  symmetric, the projection method should be performed in a sequential manner. Each device sequentially calculates its projected aggregation weights. 

Take device \(i\) as an example. The aggregation weights, \(w_{ij}, \forall j < i\) should be equal to the projected result \(w_{ji}\)  from  the preceding device \(j\) to meet the symmetry requirement, i.e.,
\begin{align}
	w_{ij}=w_{ji}, \forall j<i.\label{222}
\end{align}

Define ${{\ww}_i}=[w_{ii+1},w_{ii+2},\cdots,w_{iM}]^\mathrm{T} \in \mathbb{R}^{m_i}$, $m_i=M-i$ and $l_i=1-\sum_{j=1}^{i-1}w_{ij}$.
For  aggregation weights \(w_{ij}, \forall j > i\), the projection result is obtained from the optimal conditions of the following  problem. 
\begin{subequations}\label{ppro}
	\begin{align}
		&\min_{\qq} \quad \norm{\qq-{{\ww}_i}}^2\\
		& ~~\text{s.t.}\quad \1^{\mathrm{T}}\qq\leq l_i,\qq\succeq0.
	\end{align}
\end{subequations}

Problem \eqref{ppro} is  a convex problem. By introducing Lagrange multipliers $\boldsymbol{\lambda}\in \mathbb{R}^{m_i}$ for  inequality constrain $\qq\succeq0$ and ${\nu}$ for  equality constraint $\1^{\mathrm{T}}\qq-l_i=0$,  the Karush–Kuhn–Tucker (KKT) conditions for the optimal primal and dual variables $\qq^\star,\boldsymbol{\lambda}^\star,\nu^\star$ are
\begin{subequations}
\begin{align}
	&\qq^\star\succeq 0,~\1^{\mathrm{T}}\qq^\star-l_i\leq 0,\\
	&\boldsymbol{\lambda}^\star\succeq 0,~\nu^\star\geq 0,\\
	&\boldsymbol{\lambda}^\star_i\qq^\star_i=0, ~\nu^\star(\1^{\mathrm{T}}\qq^\star-1)=0,~ i=1,\cdots,m_i,\\
	&2(\qq^\star_j-({{\ww}_i})_j)-\boldsymbol{\lambda}^\star_j+\nu^\star=0, j=1,\cdots,m_i.
\end{align}
\end{subequations}
By eliminating  dual variables $\boldsymbol{\lambda}^\star$, we obtain the equivalent optimality condition
\begin{subequations}
\begin{align}
	&\qq^\star\succeq 0,~\1^{\mathrm{T}}\qq^\star-l_i\leq 0,\\
	&\nu^\star\geq 0,~ \nu^\star(\1^{\mathrm{T}}\qq^\star-l_i)=0,\\
	&(2(\qq^\star_j-({{\ww}_i})_j)+\nu^\star)\qq^\star_j=0, ~j=1,\cdots,m_i,\label{22b}\\
	&2(\qq^\star_j-({{\ww}_i})_j)+\nu^\star\geq 0, ~j=1,\cdots,m_i.\label{22c}
\end{align}
\end{subequations}
From \eqref{22c}, if $\nu^\star<2({{\ww}_i})_j$, we must have $\qq^\star_j>0$. According to \eqref{22b},  $\qq^\star_j=({{\ww}_i})_j-{\nu^\star}/{2}$. Otherwise, if $\nu^\star\geq2({{\ww}_i})_j$,   $\nu^\star\geq 2({{\ww}_i})_j-2\qq^\star_j$ from \eqref{22c}. Furthermore, it is necessary to have $\qq^\star_j=0$  since $\qq^\star\succeq 0$ and the constraint \eqref{22b} holds. Therefore, we  conclude 
\begin{align}
	\qq^\star_j=\max\left\{0,({{\ww}_i})_j-{\nu^\star}/{2}\right\},~j=1,\cdots,m_i.\label{23}
\end{align}
Since $\qq^\star$ must satisfy $\1^{\mathrm{T}}\qq^\star-l_i\leq 0$,  we have $\sum_{j=1}^{m_i}\max\left\{0,({{\ww}_i})_j-{\nu^\star}/{2}\right\}\leq l_i$. Considering complementary slackness condition $\nu^\star(\1^{\mathrm{T}}\qq^\star-l_i)=0$, we can obtain the solution of $\nu^\star$ either at $\nu^\star=0$ or at  $\nu^\star$ satisfying $\sum_{j=1}^{m_i}\max\left\{0,({{\ww}_i})_j-{\nu^\star}/{2}\right\}= l_i$. Substituting  $\nu^\star$ into \eqref{23}, we obtain the projected vector $\qq^\star$.
\begin{remark}
	Our subgradient-based optimization framework is inspired by the distributed consensus optimization method in~\cite{boyd2006randomized}. 
	The proposed algorithm  can be run in a distributed manner. In step 3, the unit eigenvector can be computed based on Algorithm \ref{algOI}, which is shown to be distributed in Remark \ref{remark3}. In step 4, since each device, say device $i$, is aware of its corresponding eigenvector component $v_i$, the subgradient of the $i$-th device's aggregation weights, i.e., $ w_{ij},\forall j,$ can be computed by \eqref{17} and \eqref{199} with only local  communications. In step 5,  aggregation weights $w_{ij},\forall j$ and corresponding subgradient components $g(w_{ij}),\forall j$ are held in each device, so each device is able to calculate the updated aggregation weights. Finally, in step 6, the projection method is sequentially carried out at each device based on  \eqref{222} and \eqref{23}, and hence this step can be performed distributedly. 
\end{remark}
{
	\begin{remark}
	The proposed algorithm requires prior knowledge of  link reliability $p_{ij}$. 
	In practice, this information can be estimated locally by each node 
	based on its transmission history, for example, by computing the ratio of 
	successfully acknowledged packets using acknowledgment (ACK) and negative acknowledgment (NACK) feedback messages. 
	This provides a simple and fully distributed way to approximate $p_{ij}$ 
	without centralized assistance.
\end{remark}}

{
\subsection{Convergence and Complexity Analysis}\label{F}

We first analyze the computational complexity of the algorithm by counting the number of multiplications and additions involved in Algorithm 2. For each device, the complexity of Step 3 is \(K_{\text{max}}M\), as it involves iterative updates over \(K_{\text{max}}\) rounds, each requiring operations with \(M\) elements. The complexity of Step 4 is \(3M\) as it includes computations involving gradient updates for \(M\) elements with each accounting for 2 multiplications and 1 addition. Step 5 has a complexity of \(2M\), corresponding to the updates of the aggregation weights, and Step 6 has a complexity of \(M\), which accounts for the projection step. Since Algorithm 2 requires \(J_{\text{max}}\) iterations, the overall complexity for each device is \(\mathcal{O}(J_{\text{max}}(K_{\text{max}} + 6)M)\), which simplifies to \(\mathcal{O}(J_{\text{max}}K_{\text{max}}M)\). This demonstrates that the proposed algorithm has linear complexity with respect to the number of iterations. This property makes Algorithm 2  well-suited for deployment in distributed communication networks, where computational efficiency is critical.

According to \cite{kempe2004decentralized}, the distributed orthogonal iteration can compute
the  eigenvector in a decentralized manner with desired
accuracy through local updates.
Nevertheless, since Algorithm~1 relies on iterative numerical
computation,  approximation errors in the eigenvector are
unavoidable, which lead to inexact subgradient evaluations in
Algorithm~2.
To verify that the proposed method remains stable under such
approximation, we establish the following convergence result based on
the framework of approximate subgradient methods \cite{kiwiel2004convergence}, which is proved in Appendix \ref{app_c}.
\begin{proposition}
	\label{prop:convergence} Let $\mathbf{v}_{r}(n)$ be the exact eigenvector that lies in the
	eigenspace associated with $ \rho(\overline{\mW})$ and minimizes the distance to
	$\mathbf{v}(n)$.
	Assume that in iteration $n$, the eigenvector computed in Algorithm~1 
	satisfies $\|\mathbf{v}(n)-\mathbf{v}_r(n)\|_2^2\le\varepsilon_n$, 
	and that the stepsize is set to $\gamma_n = 1/n$.
	Then 
	the generated sequence $\{\mathbf{W}(n)\}$ obeys
	\[
	\liminf_{n} f(\mathbf{W}(n))
	\;\le\; f^* + \delta,
	\qquad
	\delta = \limsup_{n} \epsilon_n,
	\]
	where $f^*$ is the optimal value of $f$ and	$\epsilon_n = \mathcal{O}(\varepsilon_n)$.
\end{proposition}

In practice, the eigenvector accuracy, $\varepsilon_n$, decreases
 with the number of inner iterations $K_{\max}$ in
Algorithm~1~\cite{kempe2004decentralized}.
Therefore, $K_{\max}$ can be increased until the resulting
$\varepsilon_n$ keeps $\epsilon_n=\mathcal{O}(\varepsilon_n)$ below a
desired tolerance while $J_{\max}$ specifies the total number of
outer subgradient updates required to reach convergence.
}

\section{Numerical Results}
In this section we investigate the performance of the proposed distributed optimization algorithm over decentralized communication networks.

\subsection{Experiment Settings}
	We conduct the image classification task on the MNIST dataset \cite{deng2012mnist}. From the original dataset, we utilize 20,000 samples for training and 10,000 samples for validation. We implement the heterogeneous data splitting scheme described in \cite{mcmahan2017communication}. Since the MNIST dataset comprises 10 classes, we divide the devices into 10 equally sized groups, with each group assigned  dataset from a specific class.
	For the network configuration, we train a convolutional neural network (CNN) consisting of two $5 \times 5$ convolutional layers, with 10 and 20 links, respectively, each followed by $2 \times 2$ max pooling layers. This is followed by a batch normalization layer, a fully connected layer with 50 units and ReLU activation, and a final softmax output layer. The network comprises a total of 21,880 parameters. The cross-entropy loss function is utilized for training. We train this model with an NVIDIA RTX 3060Ti GPU.
	
	For the communication setup, we generate a geometric random graph to represent the  network. Specifically, we randomly place $M = 40$ devices within a  $1 \times 1$ square unit area. The probability of successful communication decays with the distance between the corresponding devices, given by $p_{ij} = p_{ji} = \exp(-r d_{ij}^v)$, where $d_{ij}$ denotes the distance between devices $i$ and  $j$, and $r$ and $v$ are adjustable parameters\footnote{{
		When a device is far from all others, i.e., $d_{ij}$ is large for all $j$, 
		the corresponding link reliabilities $p_{ij}$ approach zero, 
		which represents a device disconnection scenario.}
	}. Based on this, we generate every realization $n_{i \to j}^{(t)}$ by sampling from a Bernoulli distribution with the success probability $p_{ij}$. As for the optimization parameters, we set $K_{\max}=10^4$, $J_{\max}=10^4$, and step size $\gamma=0.01$. The results are averaged over 40 Monte Carlo trials.
	
\subsection{Validation of Surrogate Function}
To start with, we conduct experiments to analyze the impact of surrogate objective function $\rho(\overline{\mW})$ on system performance. To obtain various $\rho(\overline{\mW})$, we use the  convex optimization tool CVXPY \cite{diamond2016cvxpy} to randomly generate   matrices $\overline{\mW}$ with different $\rho(\overline{\mW})$.
	\begin{figure}[h]
	\centering                    %子图居中
	\includegraphics[width=0.8\linewidth]{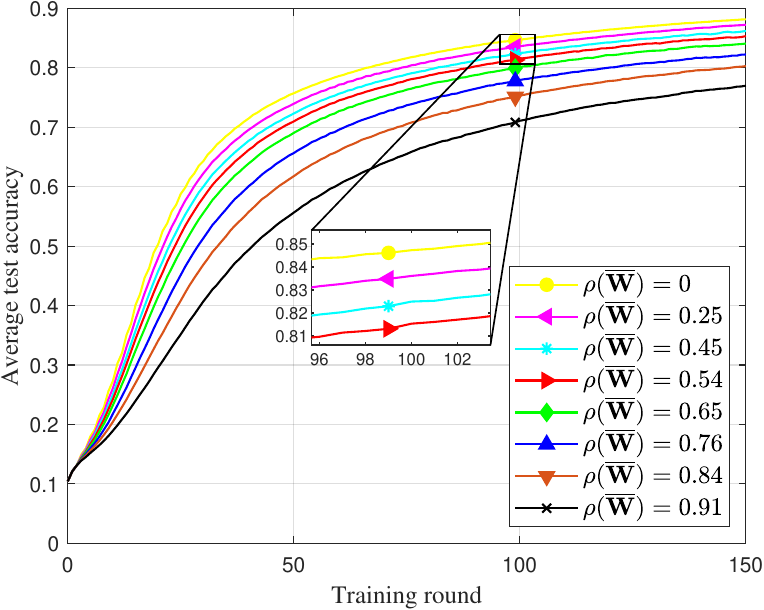}   %以pic.jpg的0.4倍大小输出
	\caption{Average test accuracy versus training round.}
	\label{rho_aver}
\end{figure}
\begin{figure}[h]
	\centering            
	\includegraphics[width=0.8\linewidth]{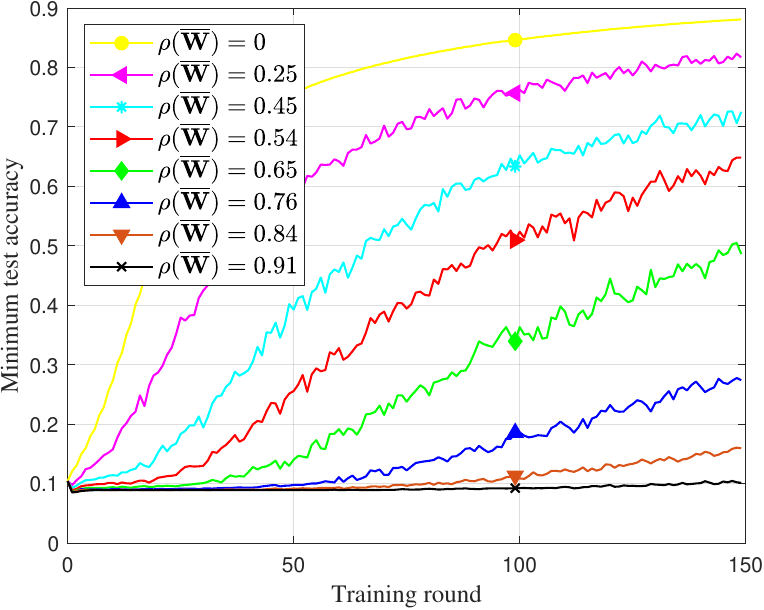}
	\caption{Minimum test accuracy versus training round.}
	\label{rho_min}
\end{figure}

In Fig.~ 2 and 3, we show the  average test accuracy (the average accuracy of all devices)  and the minimum test accuracy (the minimum accuracy among all devices) for different $\rho(\overline{\mW})$ over 150 training rounds.
From the figure, the learning accuracy gradually decreases with the increase in $\rho(\overline{\mW})$, exhibiting a monotonic relationship. This validates  the precision of the surrogate objective function derived in Section \ref{surro}. The case of $\rho(\overline{\mW}) = 0$, with its minimum accuracy and average accuracy remaining consistent, performs best in both figures. This is because $\rho(\overline{\mW}) = 0$ corresponds to the scenario where $\mW = \frac{1}{M}\1\1^{\mathrm{T}}$ and $\mP = \1\1^{\mathrm{T}}$, which is  a fully connected network with  reliable communication links. In this scenario, each device can obtain the  model aggregated from all devices (i.e., global model). This setting  is  equivalent to   traditional federated learning with error-free transmission\cite{amiri2020federated}.

Furthermore, the performance gap in the average accuracy with varying $\rho(\overline{\mW})$ is relatively small, whereas this gap becomes  significant in the minimum accuracy. This demonstrates that for  DFL deployment with larger $\rho(\overline{\mW})$,  the discrepancies between the models of different devices are huge. Therefore, it implies that  $\rho(\overline{\mW})$  has a substantial  impact on the learning and consensus performance and emphasizes the importance of the optimization.

\subsection{Convergence of Proposed Algorithm}
To validate the convergence performance of the proposed distributed algorithm, we show  objective value $\rho(\overline{\mW})$ in each subgradient iteration of Algorithm \ref{algpsa} and compare  with the value obtained by centralized optimization algorithm. In the centralized  algorithm, which assumes  a  central entity to  collect information from all edge devices and solve  problem \eqref{problem2} to optimize  aggregation weight matrix $\mW$. Since  \eqref{problem2} is a convex problem, the obtained result from centralized  algorithm is the global optimal solution. 
	\begin{figure}[h]
	\centering                    %子图居中
\includegraphics[width=0.8\linewidth]{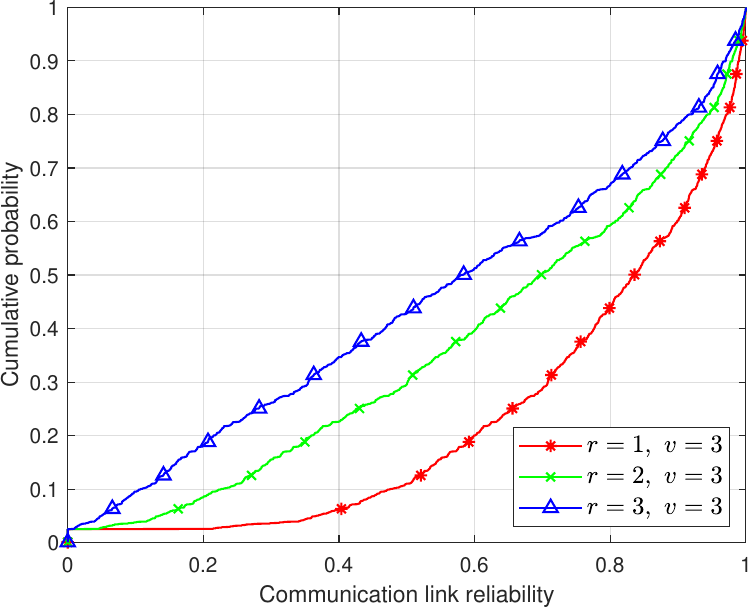}   %以pic.jpg的0.4倍大小输出
	\caption{Cumulative probability  versus distance.}
	\label{probability}
\end{figure}
	\begin{figure}[h]
	\centering            
	\includegraphics[width=0.8\linewidth]{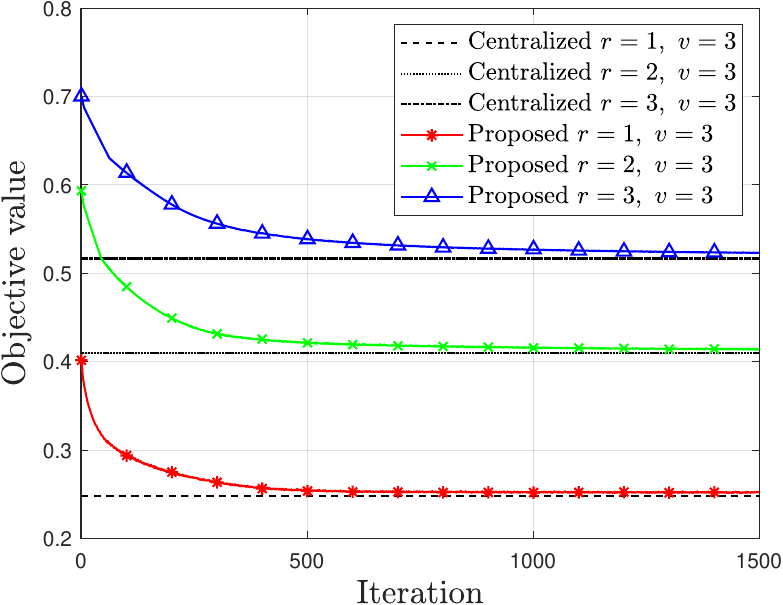}
	\caption{Convergence performance of the proposed  algorithm.}
	\label{optimization_compare}
\end{figure}

We fix the device positions  and vary the values of parameters $r$ and $v$ to acquire  results under different link reliability matrix $\mP$. 
Fig.~\ref{probability} and Fig.~\ref{optimization_compare}  show the cumulative probability function of link reliability matrix $\mP$ and the objective value of each iteration of proposed algorithm, respectively. 
In the proposed algorithm, we take  aggregation weights $\mW=\frac{1}{M}\1\1^{\mathrm{T}}$  as the initial solution. When all the communication links are reliable, i.e., $\mP=\1\1^{\mathrm{T}}$,  $\mW=\frac{1}{M}\1\1^{\mathrm{T}}$ is  the optimal solution and we have  $\rho(\overline{\mW})=0$.
From Fig.~\ref{optimization_compare},   the proposed algorithm consistently  converges to the global optimal solution given by the centralized algorithm under  different  settings.  Moreover,  as the parameter $r$ increases, the convergence value of the $\rho(\overline{\mW})$  increases and the convergence time becomes longer. As shown in Fig.~\ref{probability}, when the value of $r$ is larger, there are more communication links with low reliability and  matrix $\mP$ deviates further from the ideal case  $\mP=\1\1^{\mathrm{T}}$. In this case,  initial solution $\mW=\frac{1}{M}\1\1^{\mathrm{T}}$  is far from optimal and leads  to a longer convergence time.

\subsection{Performance under Different Settings}
In this subsection, we investigate the proposed algorithm under different   system configurations.
%By doing this, we can obtain different reliability matrices $\mP$.  

Fig.~\ref{diff_r_aver}  and  Fig.~\ref{diff_r_min}  plot the average test accuracy and the minimum test accuracy, respectively. In these figures, we set $r=2$ and change the value of $v$.  From these figures, the test accuracy decreases with the increase of parameter $v$ while and the accuracy gap is larger with respect to minimum accuracy. The reason  is that  optimized second largest (in magnitude) eigenvalue  $\rho(\overline{\mW})$  decreases (from $0.54$  when $v=2$ to $0.15$ when $v=10$)  as the value of $v$ increases.

\begin{figure}[h]
	\centering            
	\includegraphics[width=0.8\linewidth]{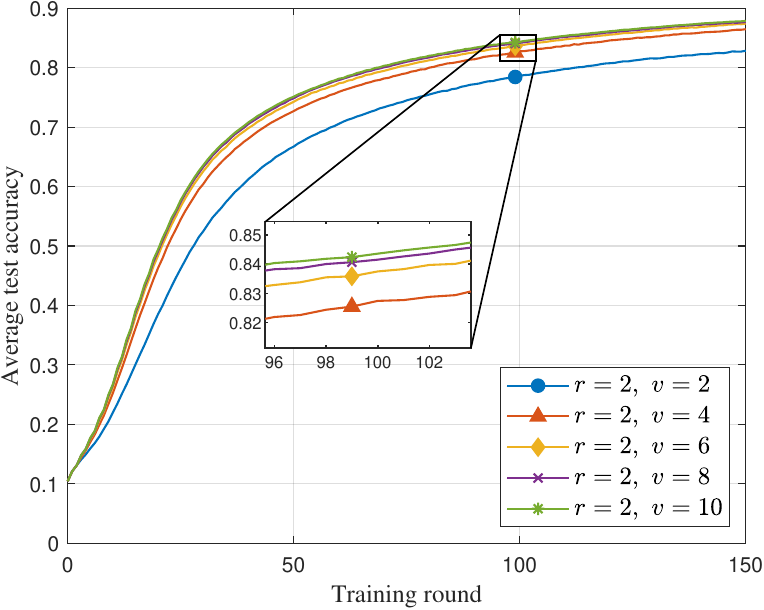}
	\caption{Average test accuracy versus training round.}
	\label{diff_r_aver}
\end{figure}
	\begin{figure}[h]
	\centering                    %子图居中
	\includegraphics[width=0.8\linewidth]{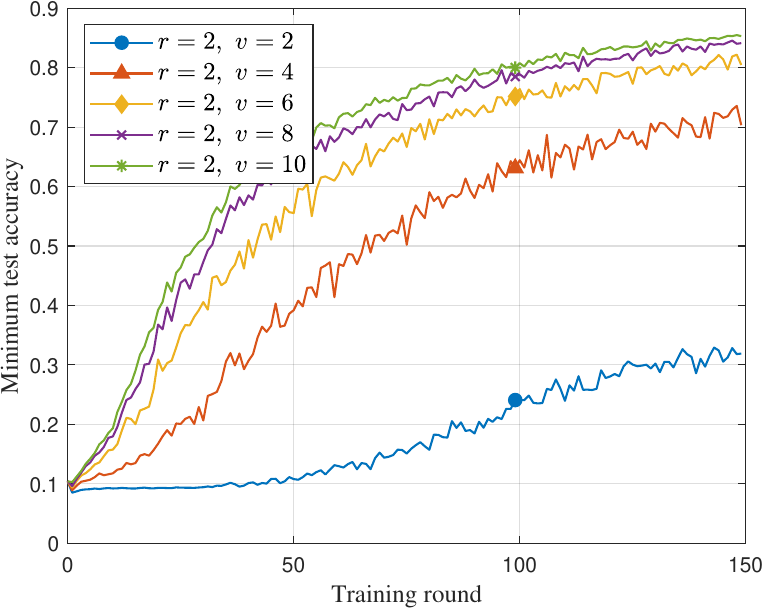}   %以pic.jpg的0.4倍大小输出
	\caption{Minimum test accuracy versus training round.}
	\label{diff_r_min}
\end{figure}

{
\begin{figure}[h]
	\centering            
	\includegraphics[width=0.8\linewidth]{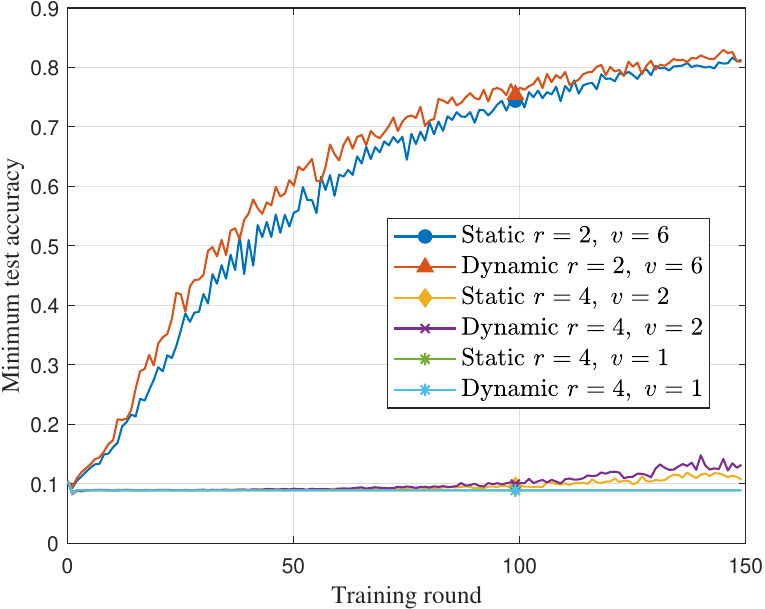}
	\caption{Average test accuracy versus training round.}
	\label{diff_r_aqqver}
\end{figure}
\begin{figure}[h]
	\centering                    %子图居中
	\includegraphics[width=0.8\linewidth]{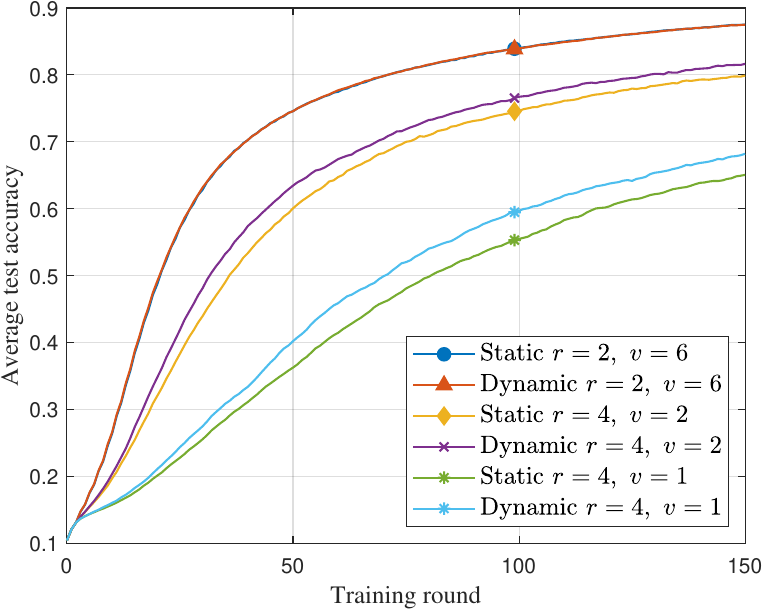}   %以pic.jpg的0.4倍大小输出
	\caption{Minimum test accuracy versus training round.}
	\label{diff_r_mqqin}
\end{figure}

Fig. \ref{diff_r_aqqver} and \ref{diff_r_mqqin} compare the performance of the decentralized federated learning (DFL) system under static and dynamic link reliability with varying \(r\) and \(v\). In the static scenario, user positions and link qualities remain constant throughout the training process, requiring a single optimization of aggregation weights. In contrast, the dynamic scenario involves random user position updates in each round, leading to varying link qualities and dynamic re-optimization of aggregation weights.
The left figure shows minimum test accuracy, while the right panel presents average test accuracy. Dynamic link conditions consistently improve both metrics compared to static links. This is because randomizing user positions prevents any device from being persistently subjected to poor channel conditions, balancing link reliability across devices over time. As a result, the dynamic scenario achieves better model aggregation and learning performance, demonstrating the method’s robustness to real-world network dynamics.
}
{
\subsection{Performance Comparison in Complex  System}

{In this subsection, we compare the proposed scheme with existing state-of-the-art schemes under a large-scale complex system. Specifically, we set the link reliability factor,  $r=4,v=2$, and the device  number $M=200$. We  conduct DFL training on the challenging FashionMNIST dataset \cite{xiao2017fashion}. Other settings are the same as that given in Section V-A. 
The followings are the  benchmarks used in comparison.}
	{\begin{itemize}
		\item \textbf{Weight design via centralized  optimization (WD via CO) \cite{9716792}}: Assume there is a central entity in the system that can collect the link reliability information, i.e.,  matrix $\mP$, from all devices. This central entity optimizes the  aggregation weights   by solving \eqref{problem2} and distributes the results to the corresponding edge devices.
		\item \textbf{Equal weight design without  reliability consideration (EW w/o RC) \cite{lin2022distributed}}: We use the average aggregation weights for each device, i.e., $\mW=\frac{1}{M}\1\1^{\mathrm{T}}$. Since this setting   consider no link reliability $\mP$, the resulting $\overline{\mW}$ does not minimize $\rho(\overline{\mW})$ in general.
		\item \textbf{Weight design with reliable communication (WD with RC)}: We set all the communication links are reliable, i.e.,  $\mP=\1\1^{\mathrm{T}}$. In this case, the optimal aggregation weights  $\mW$ is $\frac{1}{M}\1\1^{\mathrm{T}}$ since this $\mW$ makes  $\rho(\overline{\mW})=0$.
		{
	\item  \textbf{Weight design via Metropolis--Hastings algorithm (WD via MH) \cite{hastings1970monte}}: 
	Following the idea of the Metropolis--Hastings (MH) algorithm \cite{hastings1970monte}, we construct a symmetric and stochastic aggregation matrix 
	based on the link reliability, $p_{ij}$. 
	Specifically,  the aggregation weights are defined as
	\begin{align}
		w_{ij}\!=\!
		\begin{cases}
			\displaystyle \frac{p_{ij}}{\max\{d_i,d_j\}}, & i \neq j, \\[0.6em]
			1 - \sum_{k \neq i} w_{ik}, & i = j,
		\end{cases}
	\end{align}
	where the {weighted degree} of device~$i$ is given by $d_i=\sum_{k} p_{ik}$. 
This design follows the principle of the MH algorithm, ensuring both symmetry ($w_{ij}=w_{ji}$) 
and stochasticity ($\sum_j w_{ij}=1$), while continuously reflecting the reliability of communication links 
in the aggregation process.

}
\end{itemize}}

	\begin{figure}[h]{
	\centering                    %子图居中
	\includegraphics[width=0.8\linewidth]{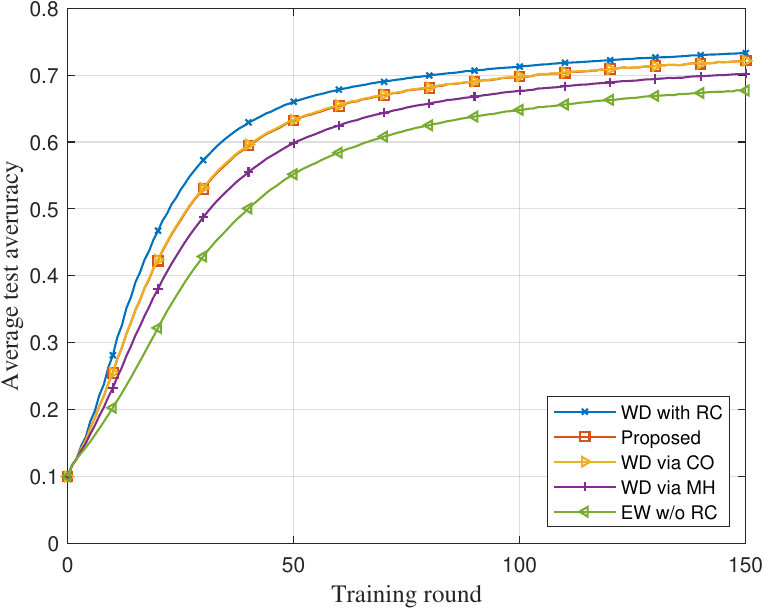}   %以pic.jpg的0.4倍大小输出
	\caption{Average test accuracy of different schemes.}
	\label{diff_sch_acc}}
\end{figure}

	\begin{figure}[h]{
	\centering                    %子图居中
	\includegraphics[width=0.8\linewidth]{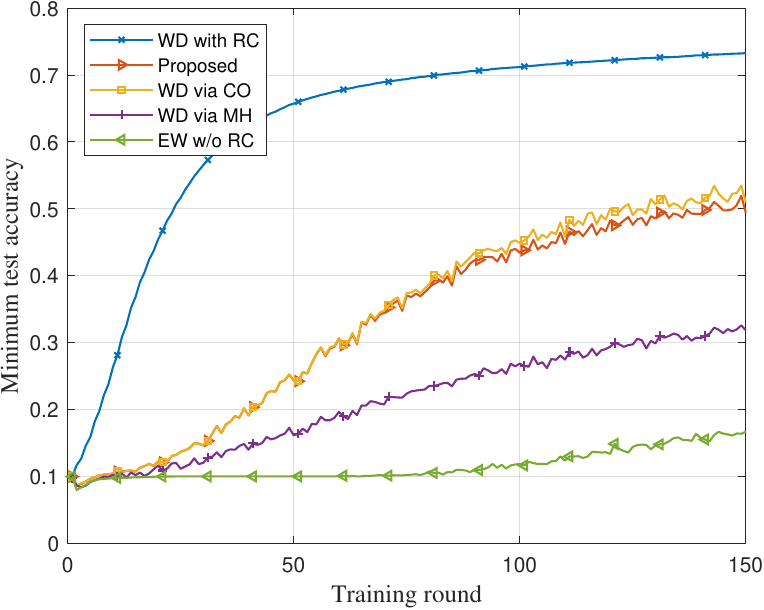}   %以pic.jpg的0.4倍大小输出
	\caption{Minimum test accuracy of different schemes.}
	\label{diff_sch_loss}}
\end{figure}

{
The average test accuracy is plotted in Fig.~\ref{diff_sch_acc} and minimum test accuracy is in Fig.~\ref{diff_sch_loss}.
From the figures,  the WD with RC scheme exhibits the best learning performance because   we set all communication links   reliable in the WD with RC and in this case the WD with RC scheme chooses  optimal aggregation weights \(\mW=\frac{1}{M}\1\1^{\mathrm{T}}\) such that $\rho(\overline{\mW})=0$. Moreover, the performance of the proposed algorithm is almost the same as that of WD via CO defined before in both figures and similar to optimal performance given by the WD with RC. This demonstrates that the proposed distributed subgradient optimization algorithm can achieve the same performance  as the centralized method that is under the assistance of a central entity. This observation can be explained by our convergence analysis. 
Since the optimization problem \eqref{problem2} is convex, the centralized method attains the {globally optimal} aggregation weights. 
Meanwhile, the projected subgradient algorithm is proven to converge to a neighborhood of the optimum (Section~\ref{F}). Therefore,  the proposed distributed realization naturally achieves performance comparable to the centralized one.
 Furthermore, the performance of the WD via WH scheme lags far behind the proposed algorithm because WD via WH heuristically designs aggregation weights based on the topology formed by the   link reliability $p_{ij}$, which does not achieve  minimal $\rho(\overline{\mW})$. Finally, the EW w/o RC scheme shows the worst learning performance because this design completely ignores the link reliability.
}
	
	\section{Conclusions}
{In this paper, we investigated distributed aggregation weight optimization for DFL. We derived an  ergodic convergence bound by capturing the impact of aggregation weights   on the learning performance over communication networks.  Based on this, we formulated a weight  optimization problem to minimize the convergence bound.  We proposed a distributed subgradient algorithm   to solve this problem. In this way, we established a completely distributed  DFL system, where optimization, communication, and learning processes are all distributed. Based on our simulation results, the proposed algorithm reached the performance of the centralized method. Future work may consider practical factors to balance performance and efficiency.}

	\appendices

\section{Proof of Proposition~\ref{pro2}}\label{appb}
We assume $\lambda\leq\frac{1}{\omega}$. Begin with $f\left(\frac{\mX^{(t+1)}\1}{M}\right)$
\begin{align}
	&\E f\left(\frac{\mX^{(t+1)}\1}{M}\right)\notag\\
	=&\E f\left(\frac{\mX^{(t)}\widehat{\mW}^{(t)}\1}{M}-\lambda \frac{\partial F(\mX^{(t)}, \boldsymbol{\xi}^{(t)})\1}{M}\right)\notag\\
%	=& \E f\left(\frac{\mX^{(t)}\1}{M}-\lambda \frac{\partial F(\mX^{(t)}, \boldsymbol{\xi}^{(t)})\1}{M}\right)\\
	\overset{(a)}{\leq}&\E f\left(\frac{\mX^{(t)}\1}{M}\right)-\lambda \E \bigg\langle\nabla f\left(\frac{\mX^{(t)}\1}{M}\right),\frac{\partial F(\mX^{(t)}, \xi^{(t)}\1}{M}\bigg\rangle \notag\\&+\frac{\omega\lambda^2}{2}\E\norm{\frac{\partial F(\mX^{(t)}, \boldsymbol{\xi}^{(t)})\1}{M}}^2\notag\\
	\overset{(b)}{=}& \E f\left(\frac{\mX^{(t)}\1}{M}\right)-\frac{\lambda}{2}\E\norm{\nabla f\left(\frac{\mX^{(t)}\1}{M}\right)}^2\notag\\+&\frac{\lambda}{2}\E\norm{\nabla f\left(\frac{\mX^{(t)}\1}{M}\right)-\frac{\partial F(\mX^{(t)}, \boldsymbol{\xi}^{(t)})\1}{M}}^2\notag\\-&\frac{\lambda}{2}\E\norm{\frac{\partial F(\mX^{(t)}, \boldsymbol{\xi}^{(t)})\1}{M}}^2+\!\frac{\omega\lambda^2}{2}\E\norm{\frac{\partial F(\mX^{(t)}, \boldsymbol{\xi}^{(t)})\1}{M}}^2\notag\\
	\overset{(c)}{\leq}&\E f\left(\frac{\mX^{(t)}\1}{M}\right)-\frac{\lambda}{2}\E\norm{\nabla f\left(\frac{\mX^{(t)}\1}{M}\right)}^2\notag\\&+\frac{\lambda}{2}{\E\norm{\nabla f\left(\frac{\mX^{(t)}\1}{M}\right)-\frac{\partial F(\mX^{(t)}, \boldsymbol{\xi}^{(t)})\1}{M}}^2},\notag
\end{align}where $(a)$ is due to $\omega$-smoothness and $\overline{\mW}\1=\1$; $(b)$ is due to $2\langle a,b\rangle=\norm{a}^2+\norm{b}^2-\norm{a-b}^2$; $(c)$ is due to the Assumption $\lambda\leq\frac{1}{\omega}$. Then, we  turn to bound the last term in the above inequality as follows.
\begin{align*}
	&\E\norm{\nabla f\left(\frac{\mX^{(t)}\1}{M}\right)-\frac{\partial F(\mX^{(t)}, \boldsymbol{\xi}^{(t)})\1}{M}}^2\\
=	&\E\norm{\nabla f\left(\frac{\mX^{(t)}\1}{M}\right)-\frac{\partial F(\mX^{(t)}, \boldsymbol{\xi}^{(t)})\1}{M}\right.\\&\left.\quad\quad-\frac{\partial f(\mX^{(t)})\1}{M}+\frac{\partial f(\mX^{(t)})\1}{M}}^2\\
%	=&\E\norm{\left(\nabla f\left(\frac{\mX^{(t)}\1}{M}\right)-\frac{\partial f(\mX^{(t)})\1}{M}\right)+\left(\frac{\partial f(\mX^{(t)})\1}{M}-\frac{\partial F(\mX^{(t)}, \boldsymbol{\xi}^{(t)})\1}{M}\right)}^2\\	
	=&\E\norm{\nabla f\left(\frac{\mX^{(t)}\1}{M}\right)-\frac{\partial f(\mX^{(t)})\1}{M}}^2\\&+\E\norm{\frac{\partial f(\mX^{(t)})\1}{M}-\frac{\partial F(\mX^{(t)}, \boldsymbol{\xi}^{(t)})\1}{M}}^2\\
	&+2\E\bigg\langle\nabla f\left(\frac{\mX^{(t)}\1}{M}\right)-\frac{\partial f(\mX^{(t)})\1}{M},\\&\quad\quad \frac{\partial f(\mX^{(t)})\1}{M}-\frac{\partial F(\mX^{(t)}, \boldsymbol{\xi}^{(t)})\1}{M}\bigg\rangle\\
	=&\E\norm{\nabla f\left(\frac{\mX^{(t)}\1}{M}\right)-\frac{\partial f(\mX^{(t)})\1}{M}}^2\\&\quad\quad+\E\norm{\frac{\partial f(\mX^{(t)})\1}{M}-\frac{\partial F(\mX^{(t)}, \boldsymbol{\xi}^{(t)})\1}{M}}^2\\
	%	&+2\E\bigg\langle\nabla f\left(\frac{\mX^{(t)}\1}{M}\right)-\frac{\partial f(\mX^{(t)})\1}{M},\E_{\xi_{k,i}}\frac{\sum_{i=1}^{n}\nabla f_i(x_{k,i}) -\sum_{i=1}^{n}\nabla F_i(x_{k,i}, \xi_{k,i})}{n}\bigg\rangle\\
	%	=&\E\norm{\nabla f\left(\frac{\mX^{(t)}\1}{M}\right)-\frac{\partial f(\mX^{(t)})\1}{M}}^2+\E\norm{\frac{\partial f(\mX^{(t)})\1}{M}-\frac{\partial F(\mX^{(t)}, \boldsymbol{\xi}^{(t)})\1}{M}}^2\\
%	{=}&\E\norm{\nabla f\left(\frac{\mX^{(t)}\1}{M}\right)-\frac{\partial f(\mX^{(t)})\1}{M}}^2\!\!\!\!+\!\!\E\norm{\frac{\sum_{i=1}^{n}\nabla f_i(x_{k,i}) -\sum_{i=1}^{n}\nabla F_i(x_{k,i}, \xi_{k,i})}{n}}^2\\
	%	=&\E\norm{\nabla f\left(\frac{\mX^{(t)}\1}{M}\right)-\frac{\partial f(\mX^{(t)})\1}{M}}^2+\frac{1}{M^2}\sum_{i=1}^{M}\E\norm{\nabla f_i(x_{k,i}) -\nabla F_i(x_{k,i}, \xi_{k,i})}^2\\
	{\leq}&\E\norm{\nabla f\left(\frac{\mX^{(t)}\1}{M}\right)-\frac{\partial f(\mX^{(t)})\1}{M}}^2+\frac{\alpha^2}{M},
\end{align*}where the last inequality is obtained  from Assumption \ref{as5}.
We then bound the second last term on the right hand side (RHS) of the above inequality as follows.
\begin{align*}
	&\E\norm{\nabla f\left(\frac{\mX^{(t)}\1}{M}\right)-\frac{\partial f(\mX^{(t)})\1}{M}}^2\\
	=&\E\norm{\frac{1}{M}\sum_{i=1}^{M}\left(\nabla f_i\left(\frac{\mX^{(t)}\1}{M}\right)-\nabla f_i(\mX^{(t)}{\ee_i})\right)}^2\\
	\leq&\frac{1}{M}\sum_{i=1}^{M}\E\norm{\nabla f_i\left(\frac{\mX^{(t)}\1}{M}\right)-\nabla f_i(\mX^{(t)}{\ee_i})}^2\\
	{\leq}&\frac{\omega^2}{M}\sum_{i=1}^{M}{\E\norm{\frac{\mX^{(t)}\1}{M}-\mX^{(t)}{\ee_i}}^2},
%	{\Xi^{(t)}_i}
\end{align*}where ${\ee_i}$ is the vector with appropriate dimensions and only the $i$-th element of ${\ee_i}$ is $1$ while all others are $0$, and the last inequality is due to the $\omega$-smoothness.
We call $\frac{1}{M}\sum_{i=1}^{M}{\E\norm{\frac{\mX^{(t)}\1}{M}-\mX^{(t)}{\ee_i}}^2}$  as the consensus error at the $t$-th round. 

Till now, it can be concluded that 
	\begin{align}
		&\E f\left(\frac{\mX^{(t+1)}\1}{M}\right)
		\leq \E f\left(\frac{\mX^{(t)}\1}{M}\right)-\frac{\lambda}{2}\E\norm{\nabla f\left(\frac{\mX^{(t)}\1}{M}\right)}^2\notag\\&+\frac{\lambda\omega^2}{2M}\sum_{i=1}^{M}{\E\norm{\frac{\mX^{(t)}\1}{M}-\mX^{(t)}{\ee_i}}^2}+\frac{\lambda\alpha^2}{2M}.\label{iteration1}
	\end{align}We proceed with bounding  $\Xi^{(t)}_i=\E\norm{\frac{\mX^{(t)}\1}{M}-\mX^{(t)}{\ee_i}}^2$.
\begin{align}	
		&\Xi^{(t)}_i=
	\E\norm{ \frac{1}{M}\left(\mX^{(t-1)}{\widehat{\mW}^{(t-1)}}\1-\lambda\left(\partial F(\mX^{(t-1)}, {\boldsymbol{\xi}}^{(t-1)})\right)\1\right)\right.\notag\\&\left.\qquad-\left(\mX^{(t-1)}{\widehat{\mW}^{(t-1)}}{\ee_i}-\lambda\left(\partial F(\mX^{(t-1)}, {\boldsymbol{\xi}}^{(t-1)})\right){\ee_i}\right)}^2\notag\\
	&=\E\norm{\frac{1}{M}\left({\mX^{(0)}\1-\sum_{i=0}^{t-1}\lambda\left(\partial F({\mX^{(i)}}, {\boldsymbol{\xi}}^{(i)})\right)\1}\right)-\right.\notag\\&\left.\left(\mX^{(0)}\!\!\prod_{m=0}^{(t-1)}\!\!\widehat{\mW}^{(m)}{\ee_i}-\sum_{j=0}^{t-1}\lambda(\partial F({\mX^{(j)}}, {\boldsymbol{\xi}}^{(j)}))\!\!\!\!\prod_{m=j+1}^{(t-1)}\!\!\!\widehat{\mW}^{(m)}{\ee_i}\right)}^2\notag\\
	&=\E\norm{\mX^{(0)}\left(\frac{\1}{M}-\prod_{m=0}^{(t-1)}\widehat{\mW}^{(m)}{\ee_i}\right)\right.-\notag\\&\left.\quad\quad\sum_{j=0}^{t-1}\lambda\left(\partial F({\mX^{(j)}}, {\boldsymbol{\xi}}^{(j)})\right)\left(\frac{\1}{M}-\prod_{m=j+1}^{(t-1)}\widehat{\mW}^{(m)}{\ee_i}\right)}^2\notag\\
	&{=}\E\norm{\sum_{j=0}^{t-1}\lambda\left(\partial F({\mX^{(j)}}, {\boldsymbol{\xi}}^{(j)})\right)\left(\frac{\1}{M}-\prod_{m=j+1}^{(t-1)}\widehat{\mW}^{(m)}{\ee_i}\right)}^2\notag\\
	&{=}\lambda^2\E\left\Vert\sum_{j=0}^{t-1}\left(\partial F({\mX^{(j)}}, {\boldsymbol{\xi}}^{(j)})-\partial f({\mX^{(j)}})+\partial f({\mX^{(j)}})\right)\right.\notag\\&\left.\quad\quad\left(\frac{\1}{M}-\prod_{m=j+1}^{(t-1)}\widehat{\mW}^{(m)}{\ee_i}\right)\right\Vert^2\notag\\	
	&\leq2\lambda^2{\E\norm{\sum_{j=0}^{t-1}\partial f({\mX^{(j)}})\left(\frac{\1}{M}-\prod_{m=j+1}^{(t-1)}\widehat{\mW}^{(m)}{\ee_i}\right)}^2}+\notag\\&2\lambda^2{\E\norm{\sum_{j=0}^{t-1}\!\left(\!\partial F({\mX^{(j)}}, {\boldsymbol{\xi}}^{(j)})-\partial f({\mX^{(j)}})\!\right)\!\!\left(\frac{\1}{M}-\!\!\!\!\prod_{m=j+1}^{(t-1)}\!\!\!\widehat{\mW}^{(m)}{\ee_i}\right)\!\!}^2}\!\!\!, \label{Qk_bound1}
\end{align}where  we use $\mX^{(0)}=0$. We now bound the first term on the RHS of inequality \eqref{Qk_bound1}.
\begin{align}
\E&\norm{\sum_{j=0}^{t-1}\!\left(\partial F({\mX^{(j)}}, {\boldsymbol{\xi}}^{(j)})-\partial f({\mX^{(j)}})\right)\!\!\left(\frac{\1}{M}-\prod_{m=j+1}^{(t-1)}\widehat{\mW}^{(m)}{\ee_i}\right)\!}^2\notag\\
	\leq&\sum_{j=0}^{t-1}\E\norm{\left(\partial F({\mX^{(j)}}, {\boldsymbol{\xi}}^{(j)})-\partial f({\mX^{(j)}})\right)}^2\notag\\&\quad\quad\norm{\left(\frac{\1}{M}-\prod_{m=j+1}^{(t-1)}\widehat{\mW}^{(m)}{\ee_i}\right)}^2\notag\\
	\leq&\sum_{j=0}^{t-1}\E\norm{\left(\partial F({\mX^{(j)}}, {\boldsymbol{\xi}}^{(j)})-\partial f({\mX^{(j)}})\right)}^2_F\notag\\&\quad\quad\norm{\left(\frac{\1}{M}-\prod_{m=j+1}^{(t-1)}\widehat{\mW}^{(m)}{\ee_i}\right)}^2\notag\\
	\overset{(a)}{\leq}&{M\alpha^2}\sum_{j=0}^{t-1}\E\norm{\left(\frac{\1}{M}-\prod_{m=j+1}^{(t-1)}\widehat{\mW}^{(m)}{\ee_i}\right)}^2\notag\\
\overset{(b)}{=}&{M\alpha^2}\sum_{j=0}^{t-1}\E\norm{\prod_{m=j+1}^{(t-1)}\left(\frac{\1\1^{\mathrm{T}}}{M}-\widehat{\mW}^{(m)}\right){\ee_i}}^2\notag,
\end{align}
where $(a)$ is from Assumption \ref{as5}, and $(b)$ is from the fact that  $\widehat{\mW}^{(m)}$ is a doubly stochastic matrix. We then bound
\begin{align}
	&\E\norm{\prod_{m=j+1}^{(t-1)}\left(\frac{\1\1^{\mathrm{T}}}{M}-\widehat{\mW}^{(m)}\right){\ee_i}}^2\notag\\
%	=&\E\left\{\ee_i^{\mathrm{T}}\left(\prod_{m=j+1}^{(t-1)}{\left(\frac{\1\1^{\mathrm{T}}}{M}-\widehat{\mW}^{(m)}\right)}\right)^{\mathrm{T}}\right.\notag\\&\quad\quad\left.\left(\prod_{m=j+1}^{(t-1)}{\left(\frac{\1\1^{\mathrm{T}}}{M}-\widehat{\mW}^{(m)}\right)}\right)\ee_i\right\}\notag\\
	=&\E\left\{\ee_i^{\mathrm{T}}\left(\prod_{m=t-1}^{(j+1)}{\left(\frac{\1\1^{\mathrm{T}}}{M}-\widehat{\mW}^{(m)}\right)^{\mathrm{T}}}\right)\right.\notag\\&\quad\quad\left.\left(\prod_{m=j+1}^{(t-1)}{\left(\frac{\1\1^{\mathrm{T}}}{M}-\widehat{\mW}^{(m)}\right)}\ee_i\right)\right\}\notag\\
	=&\E\left\{\ee_i^{\mathrm{T}}\left(\prod_{m=t-1}^{(j+2)}{\left(\frac{\1\1^{\mathrm{T}}}{M}-\widehat{\mW}^{(m)}\right)^{\mathrm{T}}}\right)\right.
	\notag\\&\E\left[\left(\frac{\1\1^{\mathrm{T}}}{M}-\widehat{\mW}^{(j+1)}\right)^{\mathrm{T}}\left(\frac{\1\1^{\mathrm{T}}}{M}-\widehat{\mW}^{(j+1)}\right)\right]\notag\\&\left(\left.\prod_{m=j+2}^{(t-1)}{\left(\frac{\1\1^{\mathrm{T}}}{M}-\widehat{\mW}^{(m)}\right)}\right)\ee_i\right\}\notag\\
	=&\E\left\{\ee_i^{\mathrm{T}}\left(\prod_{m=t-1}^{(j+2)}{\left(\frac{\1\1^{\mathrm{T}}}{M}-\widehat{\mW}^{(m)}\right)^{\mathrm{T}}}\right)\right.
	\notag\\&\left(\overline{\mW^2}-\frac{\1\1^{\mathrm{T}}}{M}\right)\left(\left.\prod_{m=j+2}^{(t-1)}{\left(\frac{\1\1^{\mathrm{T}}}{M}-\widehat{\mW}^{(m)}\right)}\right)\ee_i\right\}\notag\\
	\overset{(a)}{\leq}&\lambda_1\left(\overline{\mW^2}-\frac{\1\1^{\mathrm{T}}}{M}\right)\E\left\{\ee_i^{\mathrm{T}}\left(\prod_{m=t-1}^{(j+2)}{\left(\frac{\1\1^{\mathrm{T}}}{M}-\widehat{\mW}^{(m)}\right)^{\mathrm{T}}}\right)\right.
	\notag\\&\left(\left.\prod_{m=j+2}^{(t-1)}{\left(\frac{\1\1^{\mathrm{T}}}{M}-\widehat{\mW}^{(m)}\right)}\right)\ee_i\right\}\notag\\
	\leq&\left(\lambda_1\left(\overline{\mW^2}-\frac{\1\1^{\mathrm{T}}}{M}\right)\right)^{t-j-1}\notag\\
	\overset{(b)}{=}&\left(\lambda_1\left(\overline{\mW^2}-\left(\overline{\mW^2}\right)^\infty\right)\right)^{t-j-1}\notag\\
	\leq&\left(\rho\left(\overline{\mW^2}\right)\right)^{t-j-1},\label{18}
\end{align}where  $(a)$ is from the Rayleigh quotient inequality, and $(b)$ is from the fact that $\overline{\mW^2}$ is a doubly stochastic matrix. 
%From Appendix \ref{app_b}, we have $\rho(\overline{\mW^2}){\leq}\delta(\overline{\mW})$, which yields
%\begin{align}
%	&{M\alpha^2}\sum_{j=0}^{t-1}\E\norm{\prod_{m=j+1}^{(t-1)}\left(\frac{\1\1^{\mathrm{T}}}{M}-\widehat{\mW}^{(m)}\right){\ee_i}}^2\notag
%	\\&\leq{M\alpha^2}\sum_{j=0}^{t-1}\delta(\overline{\mW})^{t-j-1}\leq\frac{M\alpha^2}{1-\delta(\overline{\mW})}
%\end{align}
 By applying  the  analysis in \cite{lian2017can},  the second  term on the RHS of \eqref{Qk_bound1} can be eventually bounded as 
	\begin{align*}
		&\E\norm{\sum_{j=0}^{t-1}\partial f({\mX^{(j)}})\left(\frac{\1}{M}-{\prod_{m=j+1}^{(t-1)}\widehat{\mW}^{(m)}}{\ee_i}\right)}^2\nonumber\\&\leq3\sum_{j=0}^{t-1}\sum_{h=1}^{M}\E \omega^2	{\Xi^{(j)}_h}\norm{\left(\frac{\1}{M}-{\prod_{m=j+1}^{(t-1)}\widehat{\mW}^{(m)}}{\ee_i}\right)}^2
		\nonumber\\&+6\sum_{j=0}^{t-1}\left(\sum_{h=1}^{M}\E \omega^2	{\Xi^{(j)}_h}+\E \norm{\nabla f\left(\frac{{\mX^{(j)}}\1}{M}\right)\1^\top}^2\right)\\&\times\frac{\sqrt{{\rho(\overline{\mW^2})}}^{k-j-1}}{1-\sqrt{{\rho(\overline{\mW^2})}}}+\frac{9n\beta^2}{(1-\sqrt{{\rho(\overline{\mW^2})}})^2}\nonumber\\&+3\sum_{j=0}^{t-1}\E \norm{\nabla f\left(\frac{{\mX^{(j)}}\1}{M}\right)\1^\top}^2
		\norm{\left(\frac{\1}{M}-{\mW^{t-j-1}}{\ee_i}\right)}^2.
	\end{align*}
Now we are able to   bound   $\Xi^{(t)}_i$.
\begin{align*}
	&\Xi^{(t)}_i\leq12\lambda^2\sum_{j=0}^{t-1}\left(\sum_{h=1}^{M}\omega^2\E 	{\Xi^{(j)}_h}+\E \norm{\nabla f\left(\frac{{\mX^{(j)}}\1}{M}\right)\1^\top}^2\right)\\&\times\frac{\sqrt{{\rho(\overline{\mW^2})}}^{k-j-1}}{1-\sqrt{{\rho(\overline{\mW^2})}}}+\frac{2\lambda^2M\alpha^2}{1-{\rho(\overline{\mW^2})}}+\frac{18\lambda^2M\beta^2}{(1-\sqrt{{\rho(\overline{\mW^2})}})^2}\\&+6\lambda^2\sum_{j=0}^{t-1}\E \norm{\nabla f\left(\frac{{\mX^{(j)}}\1}{M}\right)\1\!^\top}^2\norm{\!\left(\frac{\1}{M}-{\prod_{m=j+1}^{(t-1)}\widehat{\mW}^{(m)}}{\ee_i}\right)\!}^2\\&+6\lambda^2\omega^2\sum_{j=0}^{t-1}\sum_{h=1}^{M}\E 	{\Xi^{(j)}_h}\norm{\left(\frac{\1}{M}-{\prod_{m=j+1}^{(t-1)}\widehat{\mW}^{(m)}}{\ee_i}\right)}^2\\
	&\leq\frac{2\lambda^2M\alpha^2}{1-{\rho(\overline{\mW^2})}}+\frac{18\lambda^2M\beta^2}{(1-\sqrt{{\rho(\overline{\mW^2})}})^2}\\&+6\lambda^2\sum_{j=0}^{t-1}\E \norm{\nabla f\left(\frac{{\mX^{(j)}}\1}{M}\right)\1^\top}^2\boldsymbol{\Omega}\left({\rho(\overline{\mW^2})}\right)\\
	&+6\lambda^2\omega^2\sum_{j=0}^{t-1}\sum_{h=1}^{M}\E	{\Xi^{(j)}_h}\boldsymbol{\Omega}\left({\rho(\overline{\mW^2})}\right),
\end{align*}
where
\begin{align*}
	\boldsymbol{\Omega}\left({\rho(\overline{\mW^2})}\right)={\left(\rho(\overline{\mW^2})\right)}^{k-j-1}+\frac{2\sqrt{{\rho(\overline{\mW^2})}}^{k-j-1}}{1-\sqrt{{\rho(\overline{\mW^2})}}}.
\end{align*}
We next bound the consensus error $\frac{1}{M}{\sum_{i=1}^{M}\Xi^{(t)}_i}$ as
\begin{align*}
&\frac{1}{M}\!\sum_{i=1}^{M}{\E\!\norm{\frac{\mX^{(t)}\1}{M}-\mX^{(t)}{\ee_i}}^2}
	\!\!\!\leq\!\!\frac{2\lambda^2M\alpha^2}{1-{\rho(\overline{\mW^2})}}+\!\frac{18\lambda^2M\beta^2}{(1-\sqrt{{\rho(\overline{\mW^2})}})^2}\\&+6\lambda^2\sum_{j=0}^{t-1}\E \norm{\nabla f\left(\frac{{\mX^{(j)}}\1}{M}\right)\1^\top}^2\boldsymbol{\Omega}\left({\rho(\overline{\mW^2})}\right)\\
	&+6\lambda^2\omega^2\sum_{j=0}^{t-1}\sum_{i=1}^{M}{\E\norm{\frac{\mX^{(j)}\1}{M}-\mX^{(j)}{\ee_i}}^2}\boldsymbol{\Omega}\left({\rho(\overline{\mW^2})}\right).
\end{align*}
The consensus error appears on the both sides of the above inequality. By summing the above inequality from $t=0$ to $T-1$, rearranging the summation, and taking relaxation, the overall bound for the consensus error is 
\begin{align}
	&\frac{1}{M}\sum_{t=0}^{T-1} \sum_{i=1}^{M}{\E\norm{\frac{\mX^{(t)}\1}{M}-\mX^{(t)}{\ee_i}}^2}\notag\\&\leq\frac{2\lambda^2}{(1-{\rho(\overline{\mW^2})})\left(1-\frac{18M\lambda^2\omega^2}{(1-\sqrt{{\rho(\overline{\mW^2})}})^2}\right)}\notag\\&\times\left(M\alpha^2T+9M\beta^2T+9\sum_{t=0}^{T-1}\E \norm{\nabla f\left(\frac{\mX^{(t)}\1}{M}\right)\1^\top}^2\right).\label{14}
\end{align}

Summing the  inequality \eqref{iteration1} from $t=0$ to $t=T-1$ while applying \eqref{14}, we obtain
\begin{align*}\label{iterate.3}
	&\frac{1}{2}\sum_{t=0}^{T-1}\E\norm{\nabla f\left(\frac{\mX^{(t)}\1}{M}\right)}^2\leq \frac{f_0-f^*}{\lambda}
	+\frac{\lambda\omega\alpha^2T}{2M}\\&+\left({\alpha^2}T+{9 \beta^2}T+{9}\sum_{t=0}^{T-1}\E \norm{\nabla f\left(\frac{\mX^{(t)}\1}{M}\right)}^2\right)G(\overline{\mW^2}),
\end{align*}
where $G(\overline{\mW^2})=\frac{M\lambda^2\omega^2}{(1-\sqrt{{\rho(\overline{\mW^2})}})^2-18M\lambda^2\omega^2}$. Rearranging the above inequality, we obtain the final convergence bound given in Theorem \ref{pro2}.

\section{Proof of Proposition~\ref{pro33}}\label{app_b}
For the second order statistic $\overline{\mW^2}$, we have
\begin{align}
	\overline{\mW^2}=(\overline{\mW})^2+\widetilde{\mW},
\end{align}
where $\widetilde{\mW}=\overline{\mW^2}-(\overline{\mW})^2$ has a variance-like form. Suppose $\rho(\overline{\mW^2})=\lambda_2(\overline{\mW^2})$,
since $\overline{\mW^2}$ is a real symmetric matrix, from the spectral stability corollary of Weyl's inequality\cite{weyl1912asymptotische}, we have
\begin{align}
	\lambda_2(\overline{\mW^2})=\lambda_2\left((\overline{\mW})^2+\widetilde{\mW}\right)\leq \lambda_2((\overline{\mW})^2)+\lambda_1(\widetilde{\mW}).\label{trans_obj}
\end{align}

Here, we see $\lambda_2((\overline{\mW})^2)+\lambda_1(\widetilde{\mW})$ is   an upper bound of $	\rho(\overline{\mW^2})$. So  $\lambda_2((\overline{\mW})^2)+\lambda_1(\widetilde{\mW})$ can be an objective function of the minimization problem of  $\rho(\overline{\mW^2})$.

Note that $\widetilde{\mW}=\overline{\mW^2}-(\overline{\mW})^2=\E\{(\widehat{\mW}^{(t)})^{\mathrm{T}}(\widehat{\mW}^{(t)})\}-\E\{(\widehat{\mW}^{(t)})\}^2=\E\{(\widehat{\mW}^{(t)})^2\}-\E\{(\widehat{\mW}^{(t)})\}^2$, which is associated with the first-order and second-order statistics of variable $\widehat{\mW}^{(t)}$. The expression of $\widetilde{\mW}$ is quite similar  to the autocovariance of a random variable. For the $(i,j)$-th element of $\widetilde{\mW}$, we have
\begin{align}
	&\widetilde{w}_{ij}=\sum_{k=1}^{M}  \left[\E\left\{\widehat{w}^{(t)}_{ki}\widehat{w}^{(t)}_{kj}\right\}-\E\left\{\widehat{w}^{(t)}_{ki}\right\}\E\left\{\widehat{w}^{(t)}_{kj}\right\}\right].\label{ee_d}
\end{align}

We  see that the $(i,j)$-th element of $\widetilde{\mW}$ is the sum of the covariance of the corresponding elements between the $i$-th and $j$-th columns of  matrix $\widehat{\mW}^{(t)}$.
%Note that the randomness of  $\widehat{\mW}^{(t)}$ is solely from $\mS^{(t)}$. 
Since the reliability of different links  is independent (from Assumption \ref{as2}), we have 
\begin{align}
	&\widetilde{w}_{ij}=2w_{ij}^2\left(p_{ij}^2-p_{ij}\right), \forall i\neq j,\\
	&\widetilde{w}_{ii}=2\sum_{k=1}^{M}w_{ki}^2\left(p_{ki}-p_{ki}^2\right), \forall i.
\end{align}
Since $|\widetilde{w}_{ii}|\geq \sum_{j=1,j\neq i}^{M}|\widetilde{w}_{ij}|$, $\widetilde{\mW}$ is a  diagonally dominant matrix. The eigenvalues of $\widetilde{\mW}$ can be estimated as its diagonal elements and the largest eigenvalue $\lambda_1(\widetilde{\mW})\approx\max\{\widetilde{w}_{ii}|~\forall i\}$. Since $0\leq p_{ij}\leq 1,\forall i,j$, we have $0\leq \widetilde{w}_{ii}\leq \frac{1}{2}\sum_{i=1}^{M}w_{ki}^2, \forall i$.

%In the  large-scale  DFL scenario where $M \to \infty$.
%
In the large-scale   systems where the number of  devices approaches infinity, there cannot not be a model of specific device  which dominates the each aggregation process. If each model shares the equal weight in  aggregation, i.e., $w_{ki}=\frac{1}{M},\forall k,i$, we have $\lim_{M\to \infty}\sum_{j=1}^{M}w_{ki}^2=\frac{1}{M}=0, \forall i$ and hence $\widetilde{w}_{ii}=0,\forall i$. Combing the above analysis, we can prove that $\lambda_1(\widetilde{\mW})\approx 0$ and $	\lambda_2(\overline{\mW^2})\leq \lambda_2((\overline{\mW})^2)$.

Similarly,  we can also prove $\rho(\overline{\mW^2})\leq \lambda_2((\overline{\mW})^2)$ for the case $\rho(\overline{\mW^2})=-\lambda_M(\overline{\mW^2})$.
Therefore,  $\lambda_2((\overline{\mW})^2)$ is an  upper bound of $\rho(\overline{\mW^2})$.
Furthermore, since $\lambda_2((\overline{\mW})^2)=(\max\{\lambda_2(\overline{\mW}),-\lambda_M(\overline{\mW})\})^2$, we can consider $\rho(\overline{\mW})=\max\{\lambda_2(\overline{\mW}),-\lambda_M(\overline{\mW})\}$
as a simplified upper bound and  turn to use it as the surrogate objective function in the optimization.

{\section{Proof of Proposition~\ref{prop:convergence}}\label{app_c}

We establish the convergence of Algorithm~2 following the framework of
approximate subgradient methods~\cite{kiwiel2004convergence}.
Given the feasible set $\mathcal{C}$ defined in~\eqref{10b},
Algorithm~2 performs the following update:
\begin{equation}
	\begin{aligned}
		\mathbf{W}(n{+}1)
		&= \Pi_{\mathcal{C}}\!\Big(\mathbf{W}(n)-\gamma_n\, g(\mathbf{W}(n))\Big),\\[-0.2em]
		g(\mathbf{W}(n)) &\in \partial_{\epsilon_n} f(\mathbf{W}(n)),
	\end{aligned}
	\label{eq:approx_update}
\end{equation}
where $\Pi_{\mathcal{C}}$ denotes the projection onto $\mathcal{C}$, and
$\partial_{\epsilon_n}f(\mathbf{W}(n))$ is the $\epsilon_n$-subdifferential defined by
\begin{align}
	&\partial_{\epsilon_n}f(\mathbf{W}(n))\notag\\
	&=\!\Big\{\!g\!:\!
	f(\mathbf{W}) \ge f(\mathbf{W}(n))
	+\langle g,\mathbf{W}-\mathbf{W}(n)\rangle-\epsilon_n,
	\forall\mathbf{W}\in\mathcal{C} \Big\}\!.\notag
\end{align}

Let $\eta_n=\tfrac{1}{2}\|g(\mathbf{W}(n))\|_F^2\gamma_n$
and $\delta_n=\eta_n+\epsilon_n$.
From~\cite{kiwiel2004convergence},
if $\sum_n\gamma_n=\infty$, then
\begin{equation}
	\liminf_{n}f(\mathbf{W}(n))
	\le f^*+\delta,\qquad
	\delta=\limsup_{n}\delta_n.
	\label{eq:app_lemma}
\end{equation}

At iteration $n$, the eigenvector $\mathbf{v}(n)$ computed by
Algorithm~1 approximates the exact eigenvector
$\mathbf{v}_r(n)$ associated with $\rho(\overline{\mathbf{W}})$ and satisfies
$\|\mathbf{v}(n)-\mathbf{v}_r(n)\|_2^2\le\varepsilon_n$.
Let $g_r(\mathbf{W}(n))$ denote the exact subgradient derived from
$\mathbf{v}_r(n)$ and $g(\mathbf{W}(n))$ the one used in the algorithm.
For any feasible $\mathbf{W}\!\in\!\mathcal{C}$,
\begin{align*}
	&\lambda_2(\mathbf{W})
	\ge \lambda_2(\mathbf{W}(n))
	+\langle g_r,\,\mathbf{W}-\mathbf{W}(n)\rangle \notag\\
	&= \lambda_2(\mathbf{W}(n))
	+\langle g,\,\mathbf{W}-\mathbf{W}(n)\rangle
	-\langle g-g_r,\,\mathbf{W}-\mathbf{W}(n)\rangle.
\end{align*}
Hence the inexactness term satisfies
\begin{align}
	\epsilon_n
	=\sup_{\mathbf{W}\in\mathcal{C}}
	\langle g-g_r,\,\mathbf{W}-\mathbf{W}(n)\rangle
	=c\,\|g-g_r\|_F^2,
	\label{eq:app_eps}
\end{align}
where $c>0$ is a scaling constant.

The $(i,j)$-th element of $g-g_r$ is
\begin{align*}
&(g-g_r)_{ij}
= p_{ij}\!\left[
( v_i-v_j )^2 - ( v_{r,i}-v_{r,j} )^2
\right]\\
&= p_{ij}\!\left[
( v_i-v_{r,i} - v_j + v_{r,j} )
( v_i-v_j + v_{r,i}-v_{r,j} )
\right].
\end{align*}
Since $|v_i-v_{r,i}|\le\sqrt{\varepsilon_n}$ for all~$i$ and the eigenvectors
are normalized ($\|\mathbf{v}\|_2=1$), one has
$|v_i-v_j|\le\sqrt{2}$, which yields
\begin{equation}
	(g-g_r)_{ij}^2
	\le 32\,p_{ij}^2\,\varepsilon_n.
	\label{eq:app_bound}
\end{equation}
Summing over all existing links gives
\begin{equation}
	\|g-g_r\|_F^2
	\le 32\,\varepsilon_n
	\sum_{(i,j)}p_{ij}^2.
	\label{eq:app_norm}
\end{equation}
Substituting~\eqref{eq:app_norm} into~\eqref{eq:app_eps} yields
\begin{equation}
	\epsilon_n
	\le 32c\,\varepsilon_n
	\sum_{(i,j)}p_{ij}^2
	= \mathcal{O}(\varepsilon_n).
	\label{eq:app_epsO}
\end{equation}

Now we choose a diminishing stepsize $\gamma_n=1/n$. From \eqref{199} we see the subgradient norm
$\|g(\mathbf{W}(n))\|_F$ remains bounded, and hence
$\eta_n\to0$.
Therefore, from~\eqref{eq:app_lemma} and~\eqref{eq:app_epsO},
\[
\liminf_{n}f(\mathbf{W}(n))
\le f^*+\delta,
\qquad
\delta=\limsup_{n}\epsilon_n,
\]
which proves Proposition~\ref{prop:convergence}.
}
	\bibliographystyle{ieeetran}
	\bibliography{IEEEabrv,mybib}
\end{document}